\documentclass[10pt,amsmath,amssymb,prx,superscriptaddress,twocolumn]{revtex4}
\usepackage{amssymb, amsmath,latexsym}
\usepackage{bm}
\usepackage[pdftex]{graphicx, color}

\usepackage{multirow}

	\usepackage{xr}

\usepackage{booktabs}

\def\QED{\mbox{\rule[0pt]{1.5ex}{1.5ex}}}
\def\endproof{\hspace*{\fill}~\QED\par\endtrivlist\unskip}

 \newenvironment{proofof}[1]{\vspace*{5mm} \par \noindent
         \quad{\it Proof of #1:\hspace{2mm}}}{\endproof}

\newtheorem{theorem}    {Theorem}
\newtheorem{lemma}[theorem]     {Lemma}
\newtheorem{corollary}[theorem]  {Corollary}
\newtheorem{definition}     {Definition}
\newtheorem{remark}{Remark}

\def\im{\mathop{\bf Im}\nolimits}

\newcommand{\half}[1]{{ \rm  h}}

\newcommand{\Tr}[0]{ {\rm{Tr}}}
\newcommand{\cH}{{\cal H}}

\def\Label#1{\label{#1}\ [\ \text{#1}\ ]\ }
\def\Label{\label}

\newcommand{\affA}{Center for Emergent Matter Science (CEMS), RIKEN, Wako, Saitama 351-0198 Japan}
\newcommand{\affB}{Guraduate School of Mathematics, Nagoya University, Nagoya 464-0814 Japan}
\newcommand{\affC}{Centre for Quantum Technology, National University of Singapore,
Singapore 117543}

\begin{document}

\title{Measurement-based Formulation of Quantum Heat Engine}

\author{Masahito Hayashi}
\affiliation{\affB}
\affiliation{\affC}
\author{Hiroyasu Tajima}
\affiliation{\affA}


\begin{abstract}
There exist two formulations for quantum heat engines that model energy transfer between two microscopic systems.
One is the semi-classical scenario, and the other is the full quantum scenario. 
The former is formulated as unitary evolution for the internal system, and is adopted by the statistical mechanics community. 
In the latter, the whole process is formulated as unitary, and is adopted by the quantum information community. 
This paper proposes a model for quantum heat engines that transfer energy 
from a collection of microscopic systems to a macroscopic system like a fuel cell.
In such a situation, the amount of extracted work is visible for a human.
For this purpose, we formulate a quantum heat engine as the measurement process whose measurement outcome is the amount of extracted work.
Under this model, we derive a suitable energy conservation law and propose a more concrete submodel.
Then, we derive a novel trade-off relation between the measurability of the amount of work extraction 
and the coherence of the internal system,
which limits the applicability of the semi-classical scenario to a heat engine transferring energy from 
a collection of microscopic systems to a macroscopic system.
\end{abstract}
\maketitle


\section{Introduction}
Thermodynamics started as a study that clarifies the upper limit of the efficiency of macroscopic heat engines \cite{Carnot} and has become a huge realm of science which covers from electric batteries \cite{Fermi} to black holes \cite{Bardeen}. 
Today, with the development of experimental techniques, the study of thermodynamics is reaching a new phase. 
The development of experimental techniques is realizing micro-machines in the laboratory \cite{nano1,nano2,demon}.
We cannot apply standard thermodynamics to these small-size heat engines as it is, because it is a phenomenological theory for macroscopic systems.
In order to study such small-size heat engines, we need to use statistical mechanical approaches.

In the statistical mechanical approach, the internal system of the heat engine can be formulated as a system obeying time-dependent Hamiltonian dynamics, which is called the classical standard formulation in this paper.
The classical standard formulation has been used since Einstein and Gibbs \cite{Ehrenfest}.
For example, Bochkov-Kuzolev \cite{Bochkov-Kuzolev1,Bochkov-Kuzolev2,Bochkov-Kuzolev3,Bochkov-Kuzolev4}
showed the second law under the standard formulation for cyclic operations, and it was extended to general operations as a corollary of the Jarzynski equality\cite{Jarzynski}. 
In this way, the classical standard formulation works well for classical heat engines, and has been adopted by the statistical mechanics community \cite{Sagawa2010,Ponmurugan2010,Horowitz2010,Horowitz2011,Sagawa2012,Ito2013}.
In the statistical mechanics community, as a quantum extension of the classical standard formulation, the work extraction process was formulated to be unitary for the internal system, which is called the semi-classical scenario in this paper. 
As an example, employing this scenario, Lenard showed the second law for cyclic processes in the quantum setting \cite{Lenard}.
Also, based on this scenario, Kurchan and Tasaki\cite{Kurchan,tasaki} gave a quantum generalization of the Jarzynski equality.
The semi-classical scenario has been adopted by the statistical mechanics community \cite{Croocks,Car1,sagawa1,jacobs,sagawa2,Funo,Morikuni,Parrondo,IE1,IE2,IE1.5,BBM1,TLH}.
Here, they assume that the time evolution of the internal system is unitary.
On the other hand, researchers in the quantum information community recently discussed unitary dynamics of the whole system including the external system storing the extracted work 
\cite{Car2,Popescu2014,oneshot2,Horodecki,oneshot1,oneshot3,Egloff,Brandao,Popescu2015,CSHO}, which is called the fully quantum scenario in this paper\footnote{%
Our classification between the semi-classical scenario and the fully quantum scenario is based on the range of the unitary dynamics of our interest.
Although the papers\cite{Car2,Popescu2014,oneshot2,Horodecki,oneshot1,oneshot3,Egloff,Brandao,Popescu2015}
discuss only states diagonal in energy basis,
the unitary dynamics of their interest covers the whole system including the external system storing the extracted work. 
Hence, we classify them as  the fully quantum scenario.}.
Although both models are different,
\r{A}berg \cite[Section II-D in Supplement]{catalyst} showed that the internal unitary dynamics in the semi-classical scenario can be realized as the approximation of the unitary of the whole system\footnote{%
To realize the approximation of the unitary of the whole system,
\r{A}berg \cite{catalyst} employed the coherence in the external system storing the extracted work. 
Also, \r{A}berg \cite{catalyst} showed that the coherence of the external system can be used repeatedly to perform coherent operations.
As was commented in \cite{BVB},
when we repeatedly use the same external system, 
the overall coherent operation has diminished accuracy, and is necessarily accompanied by an increased thermodynamic cost.}.
So, both models have succeeded in analyzing heat engines that transfer energy between two micro-systems.

\begin{figure}[htbp]
\begin{center}
\includegraphics[scale=0.4]{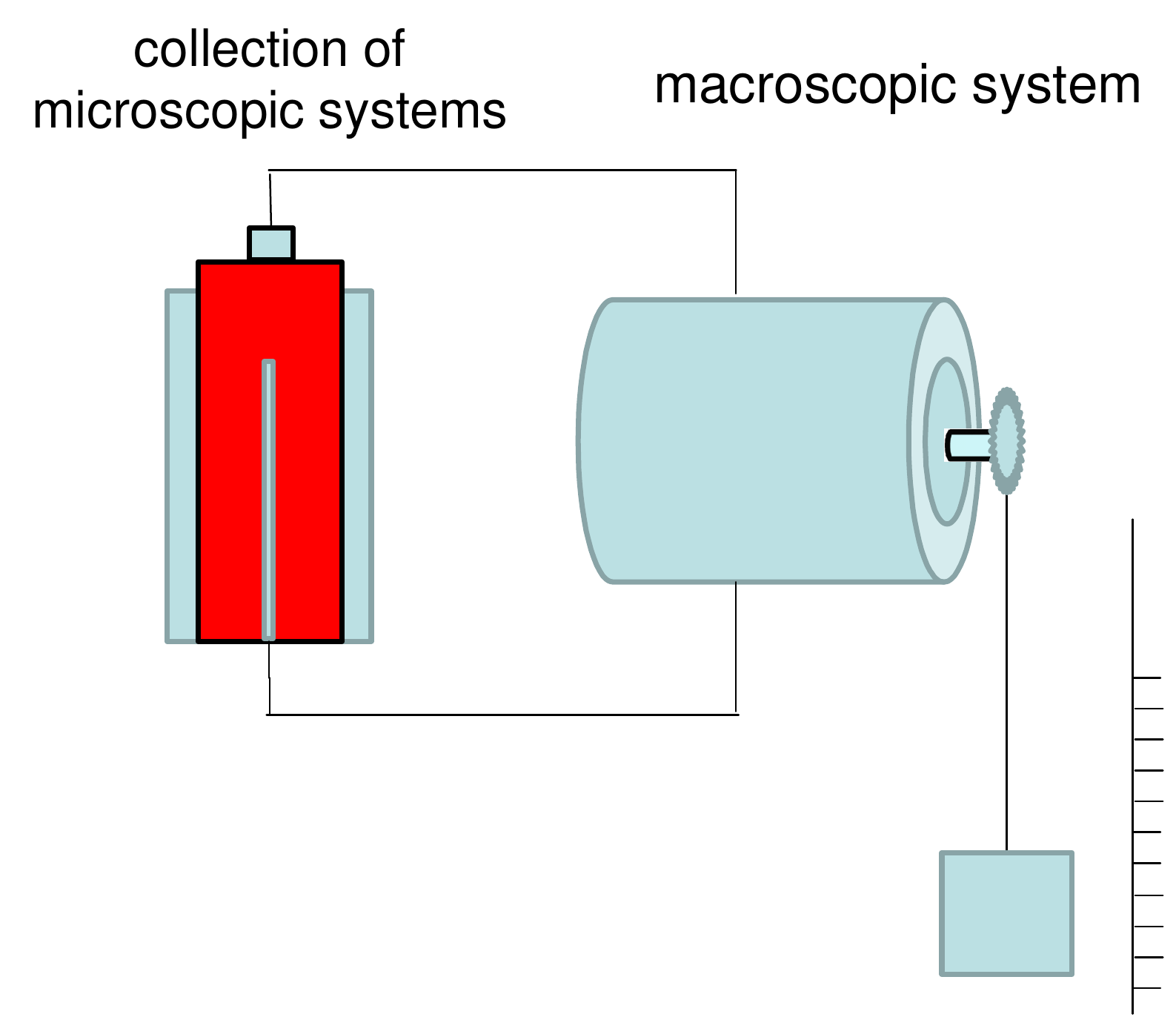}
\end{center}
\caption{Energy extraction from a collection of microscopic system to a macroscopic system}
\Label{extmeasb}
\end{figure}

To further develop quantum thermodynamics, this paper discusses heat engines transferring discernible energy 
from a collection of quantum systems to a macroscopic system, 
which is a new frontier of quantum thermodynamics.
This kind of heat engines are used in our daily life.
As a typical example, fuel cells with micro structure have been developed as solid oxide fuel cells (SOFC) recently \cite{LLA,SHTYFA,BGLAKS}.
A fuel cell has a collection of microscopic systems as an internal system and a macroscopic output power, which can be measured as the amount of extracted work by humans and effects our daily life
while the amount of extracted work does not need to be measured in the previous case as Fig. \ref{extmeasb}.
As the next topic, in order to analyze such a fuel cell with micro structure as a quantum heat engine,
it is interesting to develop a model for a heat engine that includes a measurement process to produce the output power.
That is, it is an interesting study to treat the output power as a measurement outcome, i.e., 
to formulate the heat engine as a quantum measurement process for the quantum internal system.
In this scenario, people do not measure the inside of the fuel cell, but measure alternatively the macroscopic object.

Indeed, the measurability of the amount of extracted work is a very crucial task due to the following real situation.
Consider the case when an electric bill is charged by an electric power company.
In fact, at the time of a disaster, SOFC is intended to be used as an electric power source \cite{SSYHSF}.
If the measured amount of extracted work is different from the true amount of extracted work,
the user will not pay the electric bill because he/she cannot trust the amount charged. 
To avoid such trouble, we need to precisely measure the amount of extracted work as a fundamental requirement for our model of a heat engine.
However, it is an open problem to extend quantum thermodynamics to such a case.
That is, it is our desire to formulate our model for a quantum heat engine as 
the energy extraction process that equips a quantum measurement process to output the amount of extracted work.

In the present article, we propose a general formation of quantum heat engines based on quantum measurement theory \cite{Davies1970,Ozawa1984}
as CP-work extraction, in which, 
the total system is composed of the internal system and the external system, which can be regarded as a work storage system.
As shown by Ozawa \cite{Ozawa1984}, 
such a quantum measurement process is realized by an indirect measurement process.
That is, the combination of a unitary on the whole system and the measurement of the energy on the meter system, 
which is the work storage in the current situation.
In this scenario, the initial unitary can be regarded as the fully quantum scenario \cite{Car2,Popescu2014,oneshot2,Horodecki,oneshot1,oneshot3,Egloff,Brandao,Popescu2015}.
So, the initial unitary has to satisfy the energy conservation law in the sense of the fully quantum scenario.

There are three issues to discuss about our model.
Firstly, it is not trivial to identify the energy conservation law under our CP-work extraction model.
To clarify a natural condition for energy conservation,
we consider the natural energy conservation law in the dynamics between the internal system and the quantum storage of the full-quantum model, which can be regarded as the first step of the measuring process in the indirect measurement model 
in the context of the CP-work extraction model \cite{Ozawa1984}.
When additionally we impose a natural constraint of the initial state for the quantum storage, 
we derive a very restrictive energy conservation law in an unexpected way, which will be called the level-4 energy conservation law (Theorem \ref{3-14-8L}).
That is, the restrictive condition is naturally obtained by considering 
the indirect measurement model and the existing energy conservation law in the full-quantum model \cite{oneshot2,Horodecki,oneshot1,oneshot3,Egloff,Brandao,Popescu2015}.
However, there is a case when this constraint is satisfied only partially.
Under such conditions, we derive weaker conditions as other types of energy conservation laws for CP-work extractions.
 
Second, we need a more concrete model as a natural extension of the classical standard formulation \cite{Ehrenfest}
because the above CP-work extraction is too abstract
and contains an unnatural case when the dynamics of the internal system depends on the state of the external system,
while the dynamics of the internal system is not independent of the state of the external system in the classical standard formulation \cite{Ehrenfest}. 
Fortunately, \r{A}berg \cite[Section II of Supplement]{catalyst}
discussed a concrete model that satisfies this requirement as a full-quantum model with the energy conservation law,
and we call the model the shift-invariant model
because shift-invariance guarantees the independence of the state of the external system.
However, he did not discuss the measurability of the amount of extracted work
because the main topic of his work {was} work coherence.
So, we investigate how to naturally convert the model to a CP-work extraction with the level-4 energy conservation law.
Since the shift-invariant model is obtained from a semi-classical model in a canonical way, this model can be regarded as a modification of the semi-classical model. 
Under this modification, the semi-classical model works properly when we discuss the amount of extracted work and endothermic energy.

{Third, we examine how the semi-classical scenario works approximately 
when the measurability of the amount of extracted work is imposed.
Indeed, it has been expected that the fully quantum scenario converges to the semi-classical scenario in a proper approximation
Although the semi-classical scenario assumes that the internal system evolves unitarily under a time-dependent Hamiltonian controlled by a classical external system.
This examination checks this expectation.}
To discuss this issue, we investigate 
the trade-off between the approximation of the internal unitary and the measurability of the amount of extracted work.
As a result, we derive two remarkable trade-off relations between information gain for knowing the amount of extracted work and the maintained coherence of the thermodynamic system during the work extraction process. 
These trade-off relations clarify that we can hardly know the amount of the extracted work when the time evolution of the internal system is close to unitary.

This paper is organized as follows.
Firstly, in Section \ref{s3}, we formulate work extraction as a measurement process by using CP-work extraction.
Then, we give four energy conservation laws,
level-1, level-2, level-3, and level-4 energy conservation laws, 
among which, level-4 energy conservation law is most restrictive.
Next, in Section \ref{s4}, we discuss fully quantum work extraction as a unitary process between the internal system and the work storage
as well as the energy conservation law.
Then, we discuss what kind of energy conservation laws in the CP-work extraction model 
are derived from the respective conditions for the fully quantum work extraction.
In Section \ref{s4b}, we introduce the shift invariant model 
as a modification of the semi-classical model.
We consider how well this model works as a model for a heat engine.
In Section \ref{s5}, we derive two remarkable trade-off relations between information gain for knowing the amount of extracted work and the maintained coherence of the thermodynamic system during the work
extraction process. These trade-off relations clarify that
we can hardly know the amount of extracted work
when the time evolution of the internal system is close to
unitary.

\section{Work extraction as a measurement process}\Label{s3}
In this section, we give the basic idea of our measurement-based formulation of work extraction from a quantum system to a macroscopic system. 
Let us start with standard thermodynamics; in a macroscopic heat engine, the work is given as a discernible energy change of a macroscopic work storage. 
In our quantum setting, we extract energy from a collection of quantum systems to a macroscopic system.
That is, a discernible energy change of a macroscopic work storage is caused by the effect of a collection of quantum system.
In quantum physics, such a macroscopic discernible influence caused by a quantum system
can be described only by a measurement process as Fig. \ref{extmeasb}.
Therefore, we need to formulate work extraction from the quantum system to the macroscopic system as a measurement process.

Let us formulate the above idea more concretely.
We consider a heat engine, in which, the internal system is a collection of microscopic systems
and the meter system is a macroscopic system, which can be regarded as the output system of the heat engine.
For example, a fuel battery has the fuel cells as the internal system
and the motor system as the meter system.
Hence, as the internal system, we consider a quantum system $I$, whose Hilbert space is ${\cal H}_{I}$.
We refer to the Hamiltonian of $I$ as $\hat{H}_{I}$.
The internal system $I$ usually consists of the thermodynamical system $S$ and the heat baths $\{B_{m}\}^{M}_{m=1}$, 
but we do not discuss such detailed structure of the internal system $I$, here.
Let us formulate the work extraction from $I$.
We assume that the amount of the work is indicated by a meter.(Fig.\ref{extmeas})
\begin{figure}[htbp]
\begin{center}
\includegraphics[scale=0.85]{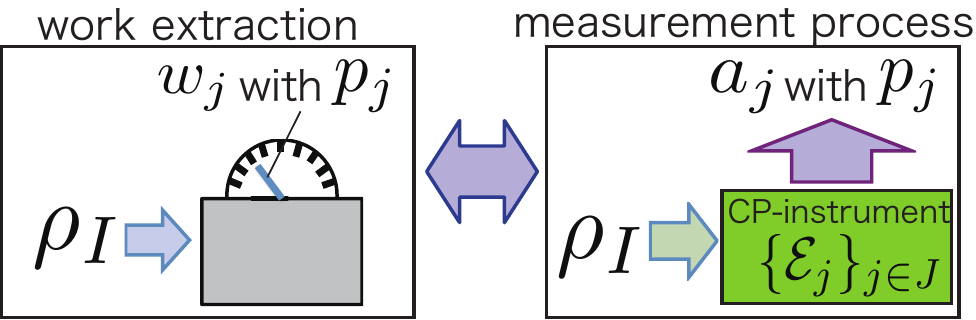}
\end{center}
\caption{Work extraction as a measurement process}
\Label{extmeas}
\end{figure}
In other words, we assume that we have an equipment to assess the amount of extracted work, and that the equipment indicates the work $w_{j}$ with the probability $p_{j}$.
In quantum mechanics, such a process that determines an indicated value $a_{j}$ with the probability $p_{j}$ is generally described as a measurement process.
Thus, we formulate work extraction from the quantum system as a measurement process\cite{Ozawa1984}.
As the minimal requirement, we demand that the average of $w_{j}$ is equal to the average energy loss of $I$ during the measurement; 
\begin{definition}[CP-work extraction]\Label{CP-normal}
Let us take an arbitrary set of a CP-instrument $\{{\cal E}_{j}\}_{j\in {\cal J}}$ and measured values $\{w_{j}\}_{j\in {\cal J}}$
satisfying the following conditions;
(1) each ${\cal E}_{j}$ is a completely positive (CP) map, 
(2) $\sum_{j}{\cal E}_{j}$ is a completely positive and trace preserving (CPTP) map, 
and (3) ${\cal J}$ is a discrete set of the outcome.
When the set $\{{\cal E}_{j},w_{j}\}_{j\in {\cal J}}$ satisfies the above condition, 
we refer to the set $\{{\cal E}_{j},w_{j}\}_{j\in {\cal J}}$ as a CP-work extraction.
\end{definition}

Here, we note that 
the measurement process $\{{\cal E}_{j}\}$ is not necessarily 
a measurement of the Hamiltonian of the internal system.
It is not difficult to treat the case where ${\cal J}$ is a continuous set, 
but to avoid mathematical difficulty, 
we consider only the case where ${\cal J}$ is discrete.
Since the heat engine needs to satisfy the conservation law of energy,
we consider the following energy conservation laws for 
a CP-work extraction $\{{\cal E}_{j},w_{j}\}_{j\in {\cal J}}$.
Firstly, we consider the weakest condition.
\begin{definition}[level-1 energy conservation law]\Label{CP-weak}
The following condition for a CP-work extraction $\{{\cal E}_{j},w_{j}\}_{j\in {\cal J}}$ 
is called the level-1 energy conservation law.
Any state $\rho_{I}$ on ${\cal H}_{I}$ satisfies 
\begin{align}
\Tr \hat{H}_I \rho_{I} = \sum_j w_j \Tr {\cal E}_j(\rho_{I})+ \sum_j \Tr \hat{H}_I {\cal E}_j(\rho_{I}),
\Label{CP1}
\end{align}
where $\hat{H}_{I}$ is the Hamiltonian of $I$.
\end{definition}

Since the level-1 energy conservation law is too weak, as explained later,
we introduce stronger conditions with an orthonormal basis $\{|x\rangle\}_x$ of $\cH_I$ such that
$|x\rangle$ is an eigenstate of the Hamiltonian $\hat{H}_I $ associated with the eigenvalue $ h_x$.
For this purpose, we introduce the spectral decomposition of $\hat{H}_I$ as $\hat{H}_I=\sum_h h P_{h}$, where $P_{h}$ is the projection to the energy eigenspace of $\hat{H}_{I}$ whose eigenvalue is $h$.

\begin{figure}[htbp]
\begin{center}
\includegraphics[scale=0.3]{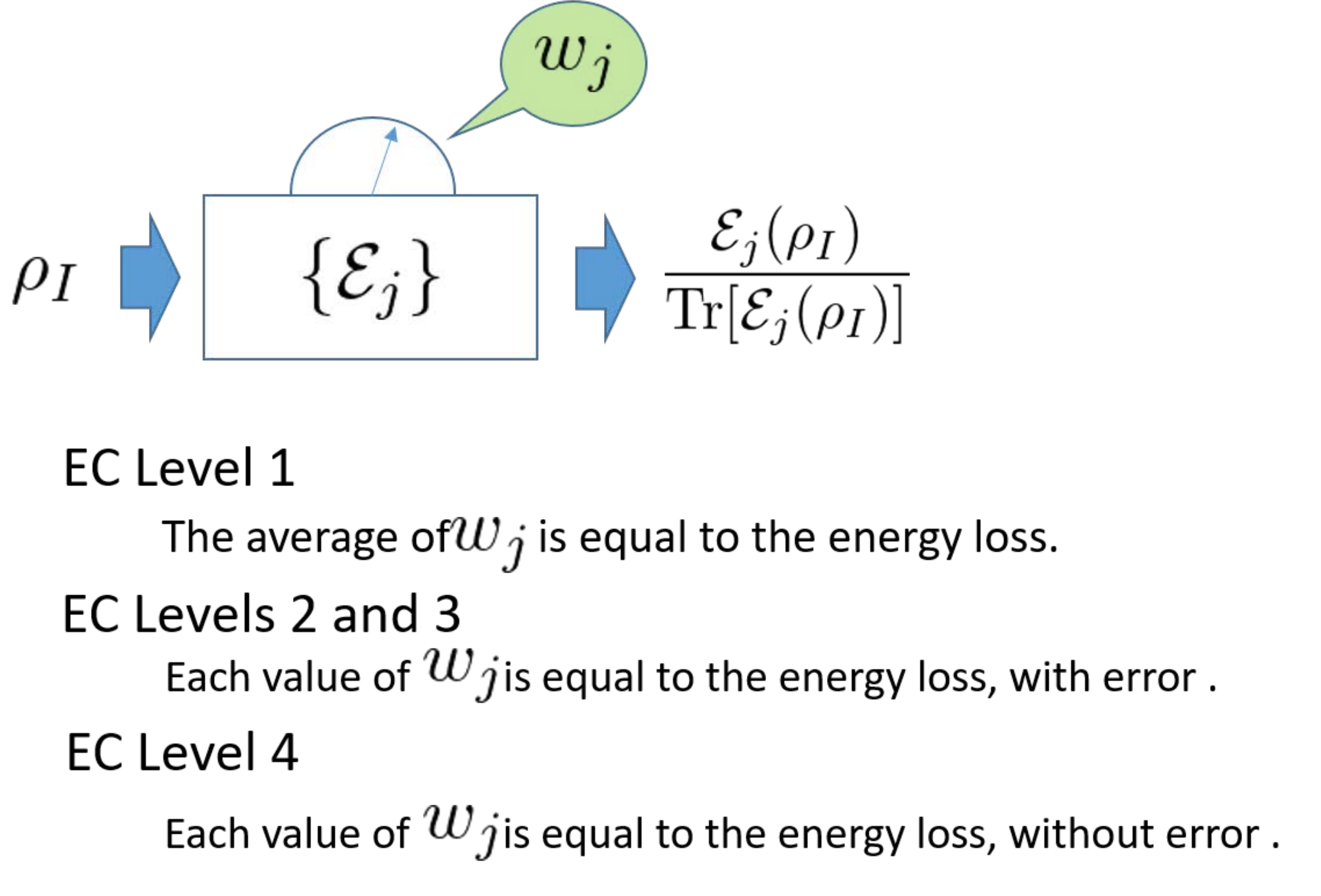}
\end{center}
\caption{The concepts of the energy conservation laws.
}
\Label{EC1fig}
\end{figure}

\begin{definition}[level-4 energy conservation law]\Label{CP-strong}
A CP-work extraction $\{{\cal E}_{j},w_{j}\}_{j\in {\cal J}}$ is called a level-4 CP-work extraction
when 
\begin{align}
{\cal E}_{j}(\Pi_{x}) 
=P_{h_{x}-w_{j}}{\cal E}_{j}(\Pi_{x})P_{h_{x}-w_{j}},
\Label{CP2}
\end{align}
for any initial eigenstate $\Pi_{x}:=|x\rangle\langle x|$.
This condition is called the level-4 energy conservation law.
\end{definition}
The meaning of \eqref{CP2} is that
the resultant state $\frac{1}{\Tr {\cal E}_{j}(\Pi_{x})}{\cal E}_{j}(\Pi_{x})$
must be an energy eigenstate with energy $h_{x}-w_{j}$
because the remaining energy in the internal system is $h_{x}-w_{j}$.
That is, the level-4 energy conservation law requires 
the conservation of energy for every possible outcome $j$.
One might consider that the level-4 energy conservation law is too strong a constraint.
However, as shown in Theorem \ref{3-14-8L}, this condition holds if and only if 
the natural energy conservation law holds as the dynamics between the internal system and the quantum storage of the full-quantum model, which can be regarded as the first step of the measuring process in the indirect measurement model (Definition \ref{FQ-normal})
and the initial state for the quantum storage is an energy eigenstate.
Further, as precisely mentioned in Lemma \ref{L5-11}, when the level-4 energy conservation law holds,
the measurement outcome precisely reflects the amount of energy lost by from the internal system.
So, such a CP-work extraction can be used for the purpose mentioned in the second paragraph of the introduction.

\begin{widetext}
\begin{figure}[htbp]
\begin{center}
\includegraphics[scale=0.3]{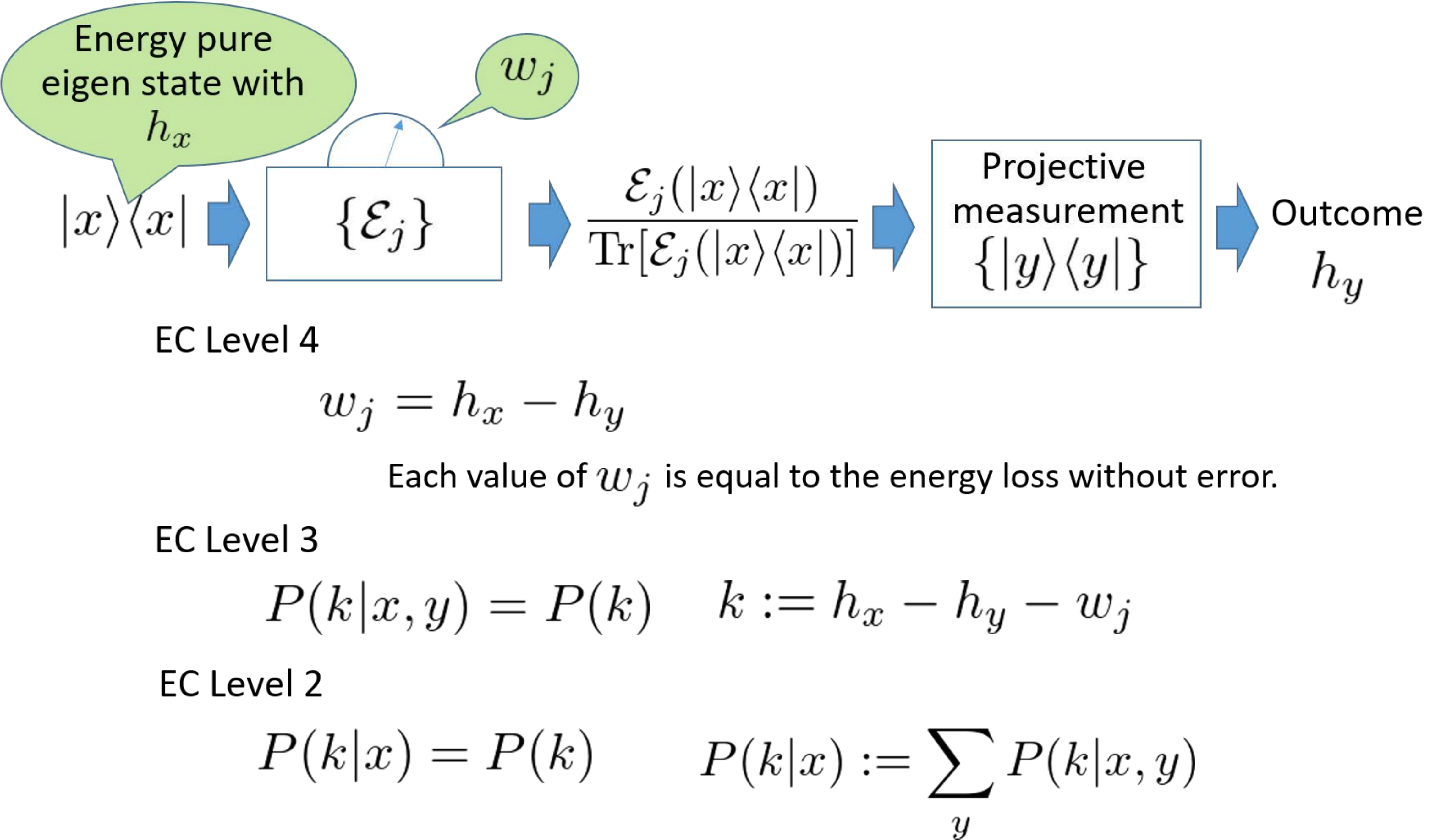}
\end{center}
\caption{The level-2, 3 and 4 energy conservation laws}
\Label{EC2fig}
\end{figure}
\end{widetext}

However, there is a possibility that the initial state of the quantum storage is not an energy eigenstate.
To characterize such a case, 
we introduce intermediate conditions in between the level-1 and level-4 energy conservation laws.
In the latter section, as Lemmas \ref{4-2-1L} and \ref{3-13-2L},
we will clarify what physical situations in the indirect model correspond to these two conservation laws. 


Here, to introduce two other energy conservation laws,
we introduce several notions for a CP-work extraction $\{{\cal E}_{j},w_{j}\}_{j\in {\cal J}}$.
Let the initial state on $I$ be an eigenstate $|x\rangle$.
After the CP-work extraction $\{{\cal E}_{j},w_{j}\}_{j\in {\cal J}}$,
we perform a measurement of $\{\Pi_{y}:=|y\rangle\langle y| \}$ 
on the resultant system ${\cal H}_I$. 
Then, we obtain the joint distribution 
$P_{JY|X}(j,y| x)$ of the two outcomes $j$ and $y$ as follows.
\begin{align}
P_{JY|X}(j,y| x)= 
\langle y| {\cal E}_{j}(\Pi_{x}) |y\rangle.
\end{align}
Then, we introduce the random variable $K:=h_{X} -h_Y - w_J$ that describes the difference between the loss of energy and the extracted energy.
So, we define the two distributions 
\begin{align}
P_{K|X}(k| x) &:= \sum_{j,y:h_{x} -h_y - w_j=k} 
P_{JY|X}(j,y| x) \\
P_{K|YX}(k| y,x) &:= \sum_{j:h_{x} -h_y - w_j=k} 
P_{J|YX}(j| y,x) ,
\end{align}
where 
\begin{align}
P_{J|YX}(j| y,x) :=
\frac{P_{JY|X}(j,y| x)}{\sum_j P_{JY|X}(j',y| x)}.
\end{align}

\begin{definition}[level-2 and 3 energy conservation laws]\Label{CP-semi}
Now, we introduce two energy conservation laws for a CP-work extraction $\{{\cal E}_{j},w_{j}\}_{j\in {\cal J}}$
when the level-1 energy conservation law holds.
When the relation $P_{K|X}(k|x)=P_{K|X}(k|x')$ holds for $k$ and $x\neq x'$,
the CP-work extraction $\{{\cal E}_{j},w_{j}\}_{j\in J}$ is called a level-2 CP-work extraction.
Similarly,
when the relation $P_{K|Y,X}(k|y,x)=P_{K|Y,X}(k|y',x')$ 
holds for $k$ and $(x,y)\neq (x',y')$,
the CP-work extraction $\{{\cal E}_{j},w_{j}\}_{j\in J}$ is called a level-3 CP-work extraction.
These conditions are called the level-2 and 3 energy conservation laws.
\end{definition}

The level-4 energy conservation law can be characterized in terms of 
the distribution $P_{K|X}$.
That is, a CP-work extraction $\{{\cal E}_{j},w_{j}\}_{j\in {\cal J}}$ is a level-4 CP-work extraction
if and only if 
\begin{align}
P_{K|X}(k| x)=\delta_{k,0} \Label{3-5-1}
\end{align}
for any initial eigenstate $|x\rangle$.
So, we find that the level-4 energy conservation law is stronger than the level-3 energy conservation law.
To investigate the property of a level-4 CP-work extraction,
we employ the pinching ${\cal P}_{\hat{H}_I}$ of the Hamiltonian $
\hat{H}_I=\sum_h h H_p$ as
\begin{align}
{\cal P}_{\hat{H}_I}(\rho):= \sum_h P_h \rho P_h .
\Label{4-8-9eqx}
\end{align}

\begin{lemma}\Label{3-14-1L}
A level-4 CP-work extraction $\{{\cal E}_{j},w_{j}\}_{j\in {\cal J}}$ 
satisfies
\begin{align}
{\cal P}_{\hat{H}_I}({\cal E}_{j}(\rho))
={\cal P}_{\hat{H}_I}(
{\cal E}_{j}({\cal P}_{\hat{H}_I}(\rho)))
=
{\cal E}_{j}({\cal P}_{\hat{H}_I}(\rho)).\Label{3-14-1eq}
\end{align}
That is, 
when we perform a measurement of an observable commuting with the Hamiltonian $\hat{H}_I$ after any level-4 CP-work extraction,
the initial state ${\cal P}_{\hat{H}_I}(\rho) $
has the same behavior as the original state $\rho$.
\end{lemma}

If we measure the Hamiltonian $\hat{H}_I$, 
we have the same result 
even if we apply the pinching ${\cal P}_{\hat{H}_I}$
before the measurement of the Hamiltonian $\hat{H}_I$.
Thus, due to Lemma \ref{3-14-1L}, 
if we measure the Hamiltonian $\hat{H}_I$ after a level-4 work extraction
$\{{\cal E}_{j},w_{j}\}_{j\in {\cal J}}$,
we have the same result 
even if we apply the pinching ${\cal P}_{\hat{H}_I}$
before the level-4 work extraction $\{{\cal E}_{j},w_{j}\}_{j\in {\cal J}}$.

\begin{proof}
We employ the Kraus representation $\{A_{j,l}\}$ of ${\cal E}_j$ 
\begin{align}
{\cal E}_j(\rho)=\sum_l A_{j,l}\rho A_{j,l}^\dagger. 
\end{align}
Then, due to the condition \eqref{CP2},
$A_{j,l}$ has the following form
\begin{align}
A_{j,l}= \sum_{h} A_{j,l,h},
\end{align}
where $A_{j,l,h}$ is a map from $\im P_h$ to $\im P_{h-w_j}$ and $\im P_h$ is the image of $P_h$. 
Thus, 
\begin{align}
P_{h}{\cal E}_j(\rho)P_{h}
=
{\cal E}_j(P_{h+w_j} \rho P_{h+w_j})
=
P_{h}{\cal E}_j(P_{h+w_j} \rho P_{h+w_j})P_{h}.
\end{align}
Taking the sum in $h$, we obtain \eqref{3-14-1eq}.
\end{proof}

Note that an arbitrary Gibbs state of a quantum system commutes with the Hamiltonian of the quantum system.
Thus, when the internal system $I$ consists of the systems in Gibbs states, a level-4 CP work extraction 
gives the energy loss of $I$ without error.

\begin{figure}[htbp]
\begin{center}
\includegraphics[scale=0.4]{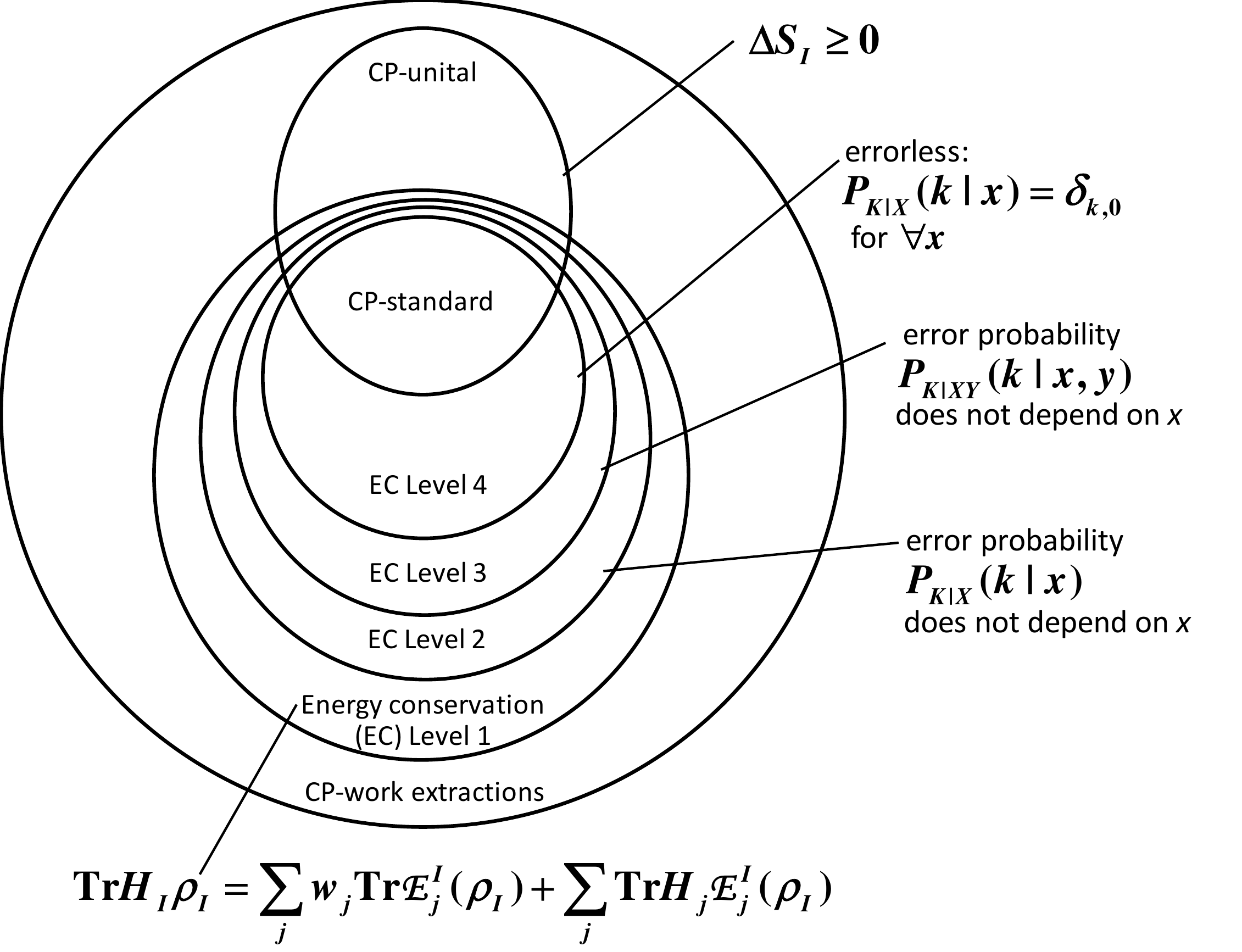}
\end{center}
\caption{A Venn diagram of the CP-work extractions}
\Label{venn1}
\end{figure}

We also consider the following condition for a CP-work extraction.
\begin{definition}[CP-unital work extraction]\Label{CP-unital}
Consider a CP-work extraction $\{{\cal E}_{j},w_{j}\}_{j\in {\cal J}}$.
When the CPTP-map $\sum_{j}{\cal E}_{j}$ is unital, namely when
\begin{align}
\sum_{j}{\cal E}_{j}(\hat{1}_{I})=\hat{1}_{I}
\Label{3-14-10eq}
\end{align}
 holds, we refer to the CP-work extraction  $\{{\cal E}_{j},w_{j}\}_{j\in {\cal J}}$ as the CP-unital work extraction.
\end{definition}
Because an arbitrary unital map does not decrease the von Neumann entropy \cite{openquantum}, the CP-unital work extraction corresponds to the class of work extractions which do not decrease the entropy of $I$.
That is,
an arbitrary $\rho_{I}$ satisfies 
\begin{eqnarray}
\Delta S_{I}&:=&S\left(\sum_{j}{\cal E}_{j}(\rho_{I})\right)-S(\rho_{I})
\ge 0,
\end{eqnarray}
where $S(\rho):=\Tr[-\rho\log\rho]$.
In contrast, we have the following characterization of the entropy of the output random variable.

\begin{lemma}\Label{4-8-1L}
Let $\{{\cal E}_{j},w_{j}\}_{j\in {\cal J}}$
be a level-4 work extraction. 
We denote the random variable describing the amount of extracted work by $W$. 
Then, for any initial state $\rho_I$ of the internal system,
the resultant entropy $S[W]$ of the system $W$ is
\begin{align}
S[W] \le 2 \log N,\Label{4-8-1eq}
\end{align}
where $N$ is the number of eigenvalues of the Hamiltonian $\hat{H}_I$
in the internal system.
\end{lemma}

One might consider that Lemma \ref{4-8-1L} is too weak to justify the unital condition.
However, as shown in Theorem \ref{3-13-5L}, 
the unital condition is a natural condition for CP-work extraction.

\begin{proof}
Due to the condition for a level-4 work extraction,
for a possible $w_j$, there exist eigenstates $x$ and $x'$ such that
$w_j=h_x-h_{x'}$.
Hence, the number of possible $w_j$ is less than $N^2$.
Thus, we obtain \eqref{4-8-1eq}.
\end{proof}

When a CP-work extraction is level-4 as well as unital,
we refer to it as a standard CP-work extraction for convenience of description
because Theorems \ref{3-14-8L} and \ref{3-13-5L} guarantee that 
these conditions are satisfied under a natural setting as illustrated in a Venn diagram (Fig.\ref{venn1}) of the CP-work extractions.

\section{Fully quantum work extraction}\Label{s4}
Next, we consider the unitary dynamics of  heat engine between the internal system $I$ and the external system $E$ which stores the extracted work from $I$.
This dynamics is essential for our CP-work extraction model as follows.
Here, the internal system $I$ is assumed to interact only with $E$, and
the external system $E$ is described by the Hilbert space  ${\cal H}_{E}$ and has the Hamiltonian $\hat{H}_{E}$.
In relation to the CP-work extraction model,
this type of description of a heat engine is given as an indirect measurement process that consists of the following two steps \cite{Ozawa1984}.
The first step is the unitary time evolution $U_{IE}$ that conserves the energy of the combined system $IE$, 
and the second step is the measurement of the Hamiltonian $\hat{H}_{E}$.
That is, the second step is given as the measurement corresponding to the spectral decomposition 
of the Hamiltonian $\hat{H}_{E}$.
While previous works \cite{Car2,Popescu2014,oneshot2,Horodecki,oneshot1,oneshot3,Egloff,Brandao,Popescu2015} have 
discussed the unitary time evolution $U_{IE}$ with a proper energy conservation law,
the relation with the CP-work extraction model was not discussed.

\begin{definition}[Fully quantum (FQ) work extraction]\Label{FQ-normal}
Let us consider an external system ${\cal H}_{E}$ with the Hamiltonian 
$\hat{H}_{E}=\sum_{j\in {\cal J}} h_{E,j} P_{E,j}$.
Then, a unitary transformation $U$ on ${\cal H}_{I}\otimes{\cal H}_{E}$
and an initial state $\rho_{E}$ of the external system ${\cal H}_{E}$
give the CP-work extraction $\{{\cal E}_j , w_j \}_{j\in {\cal J}}$ as follows.
\begin{align}
{\cal E}_j(\rho_I) &:=
\Tr_{E} U (\rho_I \otimes \rho_{E} )U^\dagger (\hat{1}_I \otimes P_{E,j}) \\
w_j &:= h_{E,j} - \Tr \hat{H}_{E} \rho_{E}.
\end{align}
Then, the quartet
${\cal F}=({\cal H}_{E}, \hat{H}_{E}, U, \rho_{E})$
is called a fully quantum (FQ) work extraction.
The above CP-work extraction $\{{\cal E}_j , w_j \}_{j\in {\cal J}}$ is simplified to 
$CP({\cal F})$.
In particular, 
the FQ-work extraction
${\cal F}$ satisfying $CP({\cal F})=\{{\cal E}_j , w_j \}_{j\in {\cal J}}$
is called a realization of the CP-work extraction $\{{\cal E}_j , w_j \}_{j\in {\cal J}}$.
\end{definition}

Here, any FQ-work extraction corresponds to a CP-work extraction.
Conversely, considering the indirect model for an instrument model,
we can show that 
there exists a FQ-work extraction
${\cal F}$ with a pure state $\rho_{E}$
for an arbitrary CP-work extraction $\{{\cal E}_j , w_j \}_{j\in {\cal J}}$ 
such that $CP({\cal F})= \{{\cal E}_j , w_j \}_{j\in {\cal J}}$\cite{Ozawa1984}\cite[Theorem 5.7]{HIKKO}.

Since the heat engine needs to satisfy the conservation law of energy,
we consider the following energy conservation laws for 
a FQ-work extraction $({\cal H}_{E}, \hat{H}_{E}, U, \rho_{E})$.

\begin{definition}[FQ energy conservation law]\Label{D7}
When a unitary $U$ is called energy conserving for the Hamiltonian $\hat{H}_I$ and $\hat{H}_{E}$
\begin{align}
[U,\hat{H}_{I}+\hat{H}_{E}]=0.\Label{FQ2}
\end{align}
Then, 
an FQ-work extraction $
{\cal F}=
({\cal H}_{E}, \hat{H}_{E}, U, \rho_{E})$ is called energy-conserving
when the unitary $U$ is energy-conserving for the Hamiltonian $\hat{H}_I$ and $\hat{H}_{E}$.
\end{definition}
The condition \eqref{FQ2} is called the FQ energy conservation law.
Note that the above condition does not depend on the choice of the initial state $\rho_{E}$ on the external system.
Indeed, the condition \eqref{FQ2} is equivalent to
the condition
\begin{align}
\Tr (\hat{H}_{I}+\hat{H}_{E}) U (\rho_I \otimes \rho_{E}) U^\dagger
&=
\Tr (\hat{H}_{I}+\hat{H}_{E}) (\rho_I \otimes \rho_{E}) 
\nonumber\\
&\enskip\enskip\enskip\enskip\hbox{ for }
\forall \rho_I \hbox{ and }\forall \rho_{E}.\Label{FQ3}
\end{align}
When we make restrictions on the state $\rho_{E}$,
the condition \eqref{FQ3} is weaker than the condition \eqref{FQ2}.
For example, when the condition \eqref{FQ3} is given with a fixed $\rho_{E}$,
the CP-work extraction $CP({\cal F})$ satisfies the level-1 energy conservation law.
However, such a restriction is unnatural, 
because such restricted energy conservation cannot recover the conventional energy conservation.
Thus, we consider the condition \eqref{FQ3} without any constraint on the state $\rho_{E}$.
Hence, we have no difference between the condition \eqref{FQ2} and the average energy conservation law \eqref{FQ3}
in this scenario.
Indeed, if we do not consider the measurement process on the external system $E$,
the model given in Definition \ref{D7} corresponds to the formulations which are used in Refs.\cite{oneshot2,Horodecki,oneshot1,oneshot3,Egloff,Brandao,Popescu2015};

Here, we discuss how to realize the unitary $U$ satisfying 
\eqref{FQ2}.
For this purpose, we prepare the following lemma.
\begin{lemma}\Label{3-13-5L'}
For an arbitrary small $\epsilon>0$
and a unitary $U$ satisfying \eqref{FQ2},
there exist a Hermitian matrix $B$ and a time $t_0>0$ such that
\begin{align}
\| B \| \le \epsilon, \quad
U= \exp(it_0 ( \hat{H}_{I}+\hat{H}_{E}+B)).\Label{4-19-5eq}
\end{align} 
\end{lemma}

\begin{proof}
Choose a Hermitian matrix $C$ such that
$\|C\|\le \pi$ and $U=\exp(i C)$.
Since $C$ and $\hat{H}_{I}+\hat{H}_{E}$ commute, 
we can choose a common basis $\{|x\rangle\}$ of $\cH_I \otimes \cH_E$ 
that diagonalizes $C$ and $\hat{H}_{I}+\hat{H}_{E}$ simultaneously.
For any $t$, we can choose a set of integers $\{n_x\}$ such that
$\|D_t\| \le \pi$,
where $D:=C- t (\hat{H}_{I}+\hat{H}_{E}) - \sum_{x}2 \pi n_x |x\rangle\langle  x|$.
Hence, the Hermitian matrix $B:=\frac{1}{t}D$ satisfies 
both conditions in \eqref{4-19-5eq} with $\epsilon = \frac{\pi}{t}$.
So, choosing $t$ large enough, we obtain the desired result.
\end{proof}

Thanks to Lemma \ref{3-13-5L'},
any unitary $U$ satisfying \eqref{FQ2} can be realized with a sufficiently long time $t$ 
by adding the small interaction Hamiltonian term $B$.
Note that the interaction $B$ does not change in $0<t<t_0$.
Thus, in order to realize the unitary $U$, we only have to turn on the interaction $B$ at $t=0$ and to turn off it at $t=t_0$.
From $t=0$ to $t=t_0$, we do not have to control the total system $IE$ time dependently.
Namely, we can realize a ``clockwork heat engine," which is programmed to perform the unitary transformation $U$ automatically.

Now, we have the following lemma.
\begin{lemma}\Label{4-2-1L}
For an energy conserving FQ-work extraction ${\cal F}=({\cal H}_{E}, \hat{H}_{E}, U, \rho_{E})$,
the CP-work extraction $CP({\cal F})$ satisfies 
the level-2 energy conservation law.
\end{lemma}

\begin{proof}
For any $j$, due to the FQ energy conservation law \eqref{FQ2},
we can choose $j'$ such that
\begin{align}
\langle x|  U^\dagger (\hat{1}_{I} \otimes P_{E,j}) |y\rangle 
=
\langle x| (\hat{1}_{I} \otimes P_{E,j'}) U^\dagger  |y\rangle .
\Label{4-2-2eq}
\end{align}
Then, the FQ energy conservation law \eqref{FQ2} implies that 
\begin{align}
& h_x -h_y- w_j =h_x -h_y-  h_{E,j}+ \Tr \hat{H}_{E} \rho_{E}
\nonumber \\
= &-h_{E,j'}+ \Tr \hat{H}_{E} \rho_{E}.\Label{4-2-1eq}
\end{align}
Hence, we can show that the distribution $P_{K|X=x}$ does not depend on $x$ as follows.
\begin{align}
&P_{K|X}(k|x) \nonumber\\
=& \sum_{j,y: h_x -h_y- w_j =k}
\Tr U (\Pi_{x} \otimes \rho_{E}) U^\dagger 
(\Pi_{y} \otimes P_{E,j}) \nonumber\\
=&
 \sum_{
 \substack{ j',y:\\ h_{E,j'}=-k+ \Tr \hat{H}_{E} \rho_{E}}
}
\Tr U (\Pi_{x} \otimes \rho_{E} )
(\hat{1}_{I} \otimes P_{E,j'})
 U^\dagger 
(\Pi_{y} \otimes \hat{1}_{E} ) \nonumber\\
\stackrel{(a)}{=} &
 \sum_{j': h_{E,j'}=-k+ \Tr \hat{H}_{E} \rho_{E}}
\Tr U (\Pi_{x} \otimes \rho_{E})
(\hat{1}_{I} \otimes P_{E,j'})
 U^\dagger 
(\hat{1}_{I} \otimes \hat{1}_{E} ) \nonumber\\
=&
 \sum_{j': h_{E,j'}=-k+ \Tr \hat{H}_{E} \rho_{E}}
\Tr (\Pi_{x} \otimes \rho_{E}) (\hat{1}_{I} \otimes P_{E,j'})
\nonumber\\
=&
 \sum_{j': h_{E,j'}=-k+ \Tr \hat{H}_{E} \rho_{E}}
\Tr_{E} \rho_{E} P_{E,j'} ,\Label{3-14-19eq}
\end{align}
which does not depend on $x$,
where $(a)$ follows from the combination of \eqref{4-2-2eq} and \eqref{4-2-1eq}.
\end{proof}

The external system interacts with the macroscopic outer system before and after the work extraction.
So, it is difficult to set the initial state $\rho_{E}$ of the external system 
to a superposition of eigenstates of the Hamiltonian $\hat{H}_{E}$.
Hence, it is natural to restrict the initial state $\rho_{E}$ to be an eigenstate of the Hamiltonian $\hat{H}_{E}$.
More generally, we restrict the initial state so that 
the support of the initial state $\rho_{E}$ 
belongs to an eigenspace of the Hamiltonian $\hat{H}_{E}$.



Now, we have the following theorem.
\begin{theorem}\Label{3-14-8L}
Let ${\cal F}=({\cal H}_{E}, \hat{H}_{E}, U, \rho_{E})$ 
be an energy-conserving FQ-work extraction.
Then, 
the support of the initial state $\rho_{E}$ 
belongs to an eigenspace of the Hamiltonian $\hat{H}_{E}$
if and only if 
the CP-work extraction $CP({\cal F})$ satisfies the level-4 energy conservation law.
\end{theorem}

\begin{proof}
The support of the initial state $\rho_{E}$ 
belongs to an eigenspace of the Hamiltonian $\hat{H}_{E}$.
Due to the assumption 
the probability 
$\Tr_{E} \rho_{E} P_{E,j'} $ 
takes a non-zero value
only in the case when $h_{E,j'}=\Tr \hat{H}_{E} \rho_{E}$.
Due to \eqref{3-14-19eq}, 
the above condition is equivalent to 
the condition that
the probability 
$P_{K|X}(k|x) $
has non-zero value only when $k=0$.
Hence, we obtain the desired equivalence relation.
\end{proof}

Finally, we have the following characterization of
the entropy of the external system $E$.

\begin{lemma}\Label{4-8-2L}
Let ${\cal F}=({\cal H}_{E}, \hat{H}_{E}, U, \rho_{E})$ 
be an energy-conserving FQ-work extraction.
We assume that 
$\rho_E$ is a pure eigenstate of $\hat{H}_E$ and
that the external system ${\cal H}_E$ has non-degenerate Hamiltonian $\hat{H}_E$, i.e., $\hat{H}_E=\sum_j h_j \Pi_{j}$ where $\Pi_{j}:=|j\rangle\langle j|$.
Then, the entropy of the final state in the external system
is
\begin{align}
S(\Tr_I U (\rho_I \otimes \rho_E) U^\dagger)
\le \log 2 N,
\Label{4-8-3eq}
\end{align}
where $N$ is the number of eigenvalues of the Hamiltonian $\hat{H}_I$
in the internal system.
\end{lemma}

\begin{proof}
Due to Lemma \ref{3-14-1L}, this FQ-work extraction
generates a CP-work extraction satisfying the level-4 condition.
So, Lemma \ref{4-8-1L} guarantees \eqref{4-8-1eq}.
Since $\hat{H}_E$ is non-degenerate,
the random variable $W$ given in Lemma \ref{4-8-1L} satisfies 
\begin{align}
S[W]= &S(\sum_j 
\Pi_{j} 
\Tr_I U (\rho_I \otimes \rho_E) U^\dagger 
\Pi_{j} ) \nonumber \\
\ge &
S(\Tr_I U (\rho_I \otimes \rho_E) U^\dagger).
\Label{4-8-2eq}
\end{align}
The combination of \eqref{4-8-1eq} of Lemma \ref{4-8-1L} and \eqref{4-8-2eq}
yields \eqref{4-8-3eq}.
\end{proof}

\section{Shift-invariant model}\Label{s4b}
The above CP work extraction and full quantum work extraction are too abstract in comparison with  the classical standard formulation.
Also, these models contain the case when the dynamics of the internal system depends on the state of the external system,
which seems unnatural.
To discuss this issue, we firstly recall the classical standard formulation:
\begin{description}
\item[\bf Classical standard formulation \cite{Ehrenfest,Bochkov-Kuzolev1,Bochkov-Kuzolev2,Bochkov-Kuzolev3,Bochkov-Kuzolev4,Sagawa2010,Ponmurugan2010,Horowitz2010,Horowitz2011,Sagawa2012,Ito2013}:]
In this scenario, we consider an external agent who performs the external operation as the classical time evolution $f$ of the internal system $I$ (which usually consists of the system $S$ and the heat bath $B$). 
In this scenario, the loss of energy of the internal system can be regarded as the amount of extracted work due to the energy conservation law.
That is, when the initial state of the internal system $x$ and the Hamiltonian is given as a function $h$, 
the amount of extracted work is $h(x)-h(f(x))$.
\end{description}
Note that the dynamics of the internal system does not depend on the state of the external system in  the classical standard formulation.
To discuss its quantum extension,
we introduce a classification of Hamiltonians.
A Hamiltonian $\hat{H}_I$ is called lattice
when there is a real positive number $d$ such that
any difference $h_i-h_j$ is an integer multiple of $d$
where $\{h_i\}$ is the set of eigenvalues of $\hat{H}_I$.
When $\hat{H}_I$ is lattice, the maximum $d$ is called the lattice span of $\hat{H}_I$.
Otherwise, it is called non-lattice.
In this subsection, we assume that our Hamiltonian $\hat{H}_I$ is lattice and denote the lattice span by $h_E$.
In the lattice case, 
using the external system $E1$ with doubly-infinite Hamiltonian,
\r{A}berg \cite[Section II of Supplement]{catalyst}
proposed a model, in which, the behavior of the heat engine depends less on the initial state of the external system.
The external system $E1$ looks unphysical, because it does not have a ground state. 
When the dimension of the internal system is finite, 
he also reconstructed the property of $E1$ in a pair of harmonic oscillators\cite[Section IV-D in Supplement]{catalyst}. 
So, we employ this definition for the simplicity of mathematical use.

Although he discussed the catalytic property and the role of coherence based on this model,
he did not discuss the relation with CP work extraction model.
In particular, he did not deal with the trade-off relation between 
the coherence and the measurability of the amount of extracted work in this model
because he discussed the average of extracted work, but not the amount of the extracted work as measurement outcome.
In this subsection, we construct essentially the same model as \r{A}berg \cite{catalyst}
in a slightly different logical step in the lattice case,
and call it a shift-invariant model while he did not give a clear name.
Then, we investigate the relation with the CP work extraction model.
In the next subsection, we extend the model to the non-lattice case while he did not discuss the non-lattice case.
In the later sections, we discuss a trade-off relation.

Consider a non-degenerate external system $E1$.
Let $\cH_{E1}$ be $L^2(\mathbb{Z})$ and the Hamiltonian 
$\hat{H}_{E1}$ be $\sum_{j} h_E j |j \rangle_E ~_E\langle j|$.
We define the displacement operator 
$V_{E1}:=\sum_{j}|j+1 \rangle_E ~_E\langle j| $.


\begin{definition}[Shift-invariant unitary]\Label{memoryless}
A unitary $U$ on $\cH_I\otimes \cH_{E1}$ is called shift-invariant when
\begin{align}
U V_{E1}=V_{E1} U.
\end{align}
\end{definition}

Indeed, there is a one-to-one correspondence between 
a shift-invariant unitary on $\cH_I\otimes \cH_{E1}$ and a unitary on $\cH_I$.
To give the correspondence, 
we define an isometry $W$ from $\cH_I$ to $\cH_I\otimes \cH_{E1}$.
\begin{align}
W:=\sum_{x} \Big|-\frac{h_x}{h_E} \Big\rangle_E\otimes \Pi_{x}.
\end{align}

\begin{lemma}\Label{3-13-1L}
A shift-invariant unitary $U$ is energy-conserving
if and only if $W^\dagger U W$ is unitary.
Conversely, for a given unitary $U_I$ on $\cH_I$, 
the operator 
\begin{align}
F[U_I]:=
\sum_{j} V_{E1}^j W U_I W^\dagger V_{E1}^{-j}
\Label{3-12-10eq}
\end{align}
on $\cH_I\otimes \cH_{E1}$
is a shift-invariant and energy-conserving unitary. 
Then, we have 
\begin{align}
W^\dagger F[U_I] W=U_I.
\Label{3-12-4}
\end{align}
\end{lemma}

Notice that the RHS of \eqref{3-12-10eq} is the same as the model given by \r{A}berg \cite[(S9) of Supplement]{catalyst}.

\begin{proof}
The image of $W$ is the eigenspace of the Hamiltonian
$\hat{H}_I+ \hat{H}_{E1}$ associated with the eigenvalue $0$.
Then, we denote the projection on the above space by $P_0$.
Hence, 
the spectral decomposition of the Hamiltonian $\hat{H}_I+ \hat{H}_{E1}$ 
is $\sum_{j} h_E j V_{E1}^j P_0 V_{E1}^{-j}$.
Since the unitary satisfies the shift-invariant condition, 
the condition \eqref{FQ2} is equivalent to 
the condition $P_0 U =P_0 U$.
The latter condition holds 
if and only if $W^\dagger U W$ is unitary.

When $U_I$ is a unitary on $\cH_I$, 
$W U_I W^\dagger$ is a unitary on the image of $W$.
So, the operator $\sum_{j} V_{E1}^j W U_I W^\dagger V_{E1}^{-j}$
on $\cH_I\otimes \cH_{E1}$
is a shift-invariant and energy-conserving unitary. 
\eqref{3-12-4} follows from the constructions.
\end{proof}

Due to \eqref{3-12-4} in Lemma \ref{3-13-1L},
we find the one-to-one correspondence between 
a shift-invariant unitary on $\cH_I\otimes \cH_{E1}$ and a unitary on $\cH_I$.
When the unitary $U_I$ is written as $\sum_{x,x'}u_{x,x'}|x \rangle \langle x'|$,
the unitary $F[U_I]$ has another expression.
\begin{align}
F[U_I]= \sum_{j,x,x'}u_{x,x'}|x \rangle \langle x'|\otimes
\left.\left.\left.\left|j+\frac{h_{x'}}{h_E}-\frac{h_x}{h_E} \right\rangle_E \right._E
\right\langle j\right|.
\end{align}

To consider such a case, we impose the following condition.
\begin{definition}[shift-invariant FQ-work extraction]\Label{FQ-shift-invariant}
We call an FQ-work extraction $({\cal H}_{E}, \hat{H}_{E}, U, \rho_{E})$ 
as shift-invariant 
when the following conditions hold.
\begin{description}
\item[Condition SI1]
The external system $E$ is the non-degenerate system $E1$, 
or the composite system of the non-degenerate system $E1$ 
and a fully degenerate system $E2$.
That is, the Hamiltonian on the additional external system $E2$ is a constant.

\item[Condition SI2]
The unitary $U$ on $(\cH_I\otimes \cH_{E2})\otimes\cH_{E1}$ is shift-invariant 
and
\begin{align}
w_j = h_E j - \Tr \hat{H}_{E} \rho_{E}.
\end{align}
\end{description}
\end{definition}

We can interpret the FQ-shift-invariant work extraction as the work extraction without memory effect
when the external system is in the non-degenerate external system $\cH_{E1}$.
Let us consider the situation that we perform CP-work extractions $n$ times.
In these applications, 
the state reduction of the external system $\cH_{E1}$
is based on the projection postulate.
Let $\rho^{(1)}_{E}$ be the initial state on the external system $\cH_{E1}$, which is assumed to be a pure state.
We assume that the initial state $\rho^{(k)}_{E}$ on $\cH_{E1}$
of the $k$th CP-work extraction is
the final state of the external system of the $k-1$th work extraction. 
Other parts of the $k$th CP-work extraction are the same as 
those of the first CP-work extraction.
Hence, the $k$th CP-work extraction is
$({\cal H}_{E}, \hat{H}_{E}, \rho^{(k)}_{E}, U)$.
Generally, the FQ-work extraction $({\cal H}_{E}, \hat{H}_{E}, \rho^{(k)}_{E}, U)$ depends on the state 
$\rho^{(k)}_{E}$.
Namely, there exists a memory effect.
However, when $U$ is shift-invariant, 
the FQ-work extraction does not depend on the state 
$\rho^{(k)}_{E}$.
Then the memory effect does not exist.
So, we don't have to initialize the external system after the projective measurement on the external system.

Then, the shift-invariant FQ-work extraction can simulate the semi-classical scenario in the following limited sense.
\begin{lemma}\Label{4-10-4L}
Given an internal unitary $U_I$ and a state $\rho_I$ of the internal system $I$,
in the shift-invariant FQ-work extraction ${\cal F}=({\cal H}_{E}, \hat{H}_{E}, F[U_I], \rho_{E})$,
the average amount of extracted work is $ \Tr \rho_I \hat{H}_I- \Tr U_I \rho U_I^\dagger \hat{H}_I$.
\end{lemma}

Further, the shift-invariant FQ-work extraction yields a special class of CP-work extraction.

\begin{lemma}\Label{3-13-2L}
For an energy conserving and shift-invariant
FQ-work extraction ${\cal F}=({\cal H}_{E}, \hat{H}_{E}, U, \rho_{E})$,
the CP-work extraction $CP({\cal F})$ satisfies 
the level-3 energy conservation law.
\end{lemma}

\begin{proof}
Firstly, we consider the case when $\cH_{E}=\cH_{E1}$.
Similar to the proof of Lemma \ref{4-2-1L}, 
we have
\begin{align}
&P_{KY|X}(k,y|x) \nonumber\\
=& \sum_{j: h_x -h_y- w_j =k}
\Tr U (\Pi_{x} \otimes \rho_{E}) U^\dagger 
(\Pi_{y} \otimes |j \rangle \langle j|  ) \nonumber\\
=& \sum_{j: h_x -h_y- w_j =k}
\Tr U (\Pi_{x} \otimes \rho_{E}) U^\dagger 
(\Pi_{y} \otimes |j \rangle \langle j|^2  ) \nonumber\\
\stackrel{(a)}{=} &
 \sum_{j': h_E j'=-k+ \Tr \hat{H}_{E} \rho_{E}}
\Tr [U (\Pi_{x} \otimes 
|j'\rangle\langle j'|\rho_{E} |j'\rangle\langle j'|  )
 U^\dagger \nonumber \\
& \cdot(\Pi_{y} \otimes \hat{1}_{E} ) ] \nonumber\\
=&
|\langle y| U_I |x\rangle |^2
 \sum_{j': h_E j'=-k+ \Tr \hat{H}_{E} \rho_{E}}
\langle j'|\rho_{E} |j'\rangle ,
\end{align}
where $(a)$ can be shown by using relations similar to \eqref{4-2-2eq} and \eqref{4-2-1eq}.
Hence,
\begin{align}
P_{K|Y,X}(k|y,x) 
=
 \sum_{j': h_E j'=-k+ \Tr \hat{H}_{E} \rho_{E}}
\langle j'|\rho_{E} |j'\rangle ,\Label{3-14-16eq}
\end{align}
which does not depend on $x$ or $y$.

Next, we proceed to the general case.
Similarly, we can show that
\begin{align}
&P_{KY|X}(k,y|x) \nonumber\\
=& \sum_{j: h_x -h_y- w_j =k}
\Tr U (\Pi_{x} \otimes \rho_{E2} \otimes \rho_{E1}) U^\dagger 
(\Pi_{y} \otimes \hat{1}_{E2} \otimes |j \rangle \langle j| 
 ) \nonumber\\
=&
(\Tr U_I \Pi_{x} \otimes \rho_{E2}U_I^\dagger
\Pi_{y} \otimes \hat{1}_{E2})
 \sum_{j': h_E j'=-k+ \Tr \hat{H}_{E} \rho_{E}}
\langle j'|\rho_{E} |j'\rangle .
\end{align}
Hence, we obtain \eqref{3-14-16eq}.
\end{proof}

As a special case of Theorem \ref{3-14-8L}, we have the following lemma.
\begin{lemma}\Label{3-13-3L}
Let ${\cal F}=({\cal H}_{E}, \hat{H}_{E}, U, \rho_{E})$ 
be an energy-conserving and shift-invariant FQ-work extraction.
Then, 
the support of the initial state $\rho_{E}$ 
belongs to an eigenspace of the Hamiltonian $\hat{H}_{E}$
if and only if 
the CP-work extraction $CP({\cal F})$ satisfies the level-4 energy conservation law.
\end{lemma}

\begin{lemma}\Label{3-13-4L}
For a level-4 CP-work extraction 
$\{{\cal E}_{j},w_{j}\}_{j\in {\cal J}} $,
there exists an energy-conserving and shift-invariant
FQ-work extraction ${\cal F}$
such that 
the support of the initial state $\rho_{E}$ 
belongs to an eigenspace of the Hamiltonian $\hat{H}_{E}$
and
$CP({\cal F})=\{{\cal E}_{j},w_{j}\}_{j\in {\cal J}}$.
\end{lemma}

\begin{proof}
We make a Stinespring extension 
$(\cH_{E2},U_{IE2},\rho_{E2})$
with a projection valued measure $\{E_j\}$ on $\cH_{E2} $ of 
$\{{\cal E}_{j},w_{j}\}_{j\in {\cal J}}$ as follows.
\begin{align}
{\cal E}_{j}(\rho_I)=
\Tr_{E2} U_{IE2}
(\rho_I \otimes \rho_{E2}) U_{IE2}^\dagger 
(\hat{1}_{I} \otimes E_j),
\end{align}
where $\{E_j\}$ is a projection valued measure on ${\cal H}_{E2}$
and $\rho_E$ is a pure state.
This extension is often called the indirect measurement model, 
which was introduced by Ozawa \cite{Ozawa1984}.
Here, the Hamiltonian of $\cH_{E2}$ is chosen to be $0$. 
Next, we define the unitary $U$ on $\cH_I \otimes \cH_{E2}
\otimes \cH_{E1}$ as $U=F[U_{IE2}]$.
Here, $\cH_{E1}$ is defined as the above discussion.

Then, we define 
an energy-conserving and shift-invariant
FQ-work extraction ${\cal F}=({\cal H}_{E}, \hat{H}_{E}, U, \rho_{E})$
with the above given $U$
as
${\cal H}_{E}:= \cH_{E1} \otimes \cH_{E2}$,
$\hat{H}_{E}:= \hat{H}_{E1}$, and
$\rho_{E}:=|0\rangle \langle 0|\otimes \rho_{E2}$.
Hence, the level-4 condition implies
\begin{align}
{\cal E}_j(\rho)=\Tr_{E} U
(\rho_I \otimes |0\rangle \langle 0|\otimes \rho_{E2}) 
U^\dagger 
(\hat{1}_{IE2} \otimes \Pi_{j}).
\end{align}
\end{proof}

It is natural to consider that 
the state on the additional external system $E2$ 
does not change due to the work extraction.
Hence, we impose the following restriction for the 
the state on the additional external system $E2$. 
\begin{definition}[Stationary condition]
A shift-invariant FQ-work extraction $({\cal H}_{E}, \hat{H}_{E}, U, \rho_{E})$ is said to satisfy 
the stationary condition when the relations 
$\rho_{E}=\rho_{E1}\otimes \rho_{E2}$ and
\begin{align}
\Tr_{IE1} U (\rho_I\otimes \rho_{E}) U^\dagger=
\rho_{E2} \Label{3-14-12eq}
\end{align}
hold for any initial state $\rho_I$ on $\cH_I$.
\end{definition}

Now, we have the following theorem.
\begin{theorem}\Label{3-13-5L}
Let ${\cal F}=({\cal H}_{E}, \hat{H}_{E}, U, \rho_{E})$ 
be a stationary FQ-work extraction.
Then, 
the CP-work extraction $CP({\cal F})$ satisfies the unital condition \eqref{3-14-10eq}.
\end{theorem}

\begin{proof}
Firstly, we consider the case when $\cH_{E}=\cH_{E1}$.
To show the unital condition \eqref{3-14-10eq}, 
it is enough to show that
\begin{align}
\Tr_I |\psi_I\rangle \langle \psi_I| \Tr_{E} U (I\otimes \rho_{E}) U ^{\dagger}
=1
\end{align}
for any pure state $|\psi_I\rangle$ on $\cH_I$.
The shift-invariant property \eqref{10-24} implies that 
\begin{align}
&\Big\langle y, j'\Big|U  \Big|x,j\Big\rangle 
=
\Big\langle y, j'\Big|U V_{E1}^{j+j'} \Big|x,-j'\Big\rangle \nonumber \\
=&
\Big\langle y, j'\Big| V_{E1}^{j+j'} U \Big|x,-j'\Big\rangle 
=
\Big\langle y, -j\Big|  U \Big|x,-j'\Big\rangle .
\end{align}
When $\rho_{E}=|\psi_{E}\rangle \langle \psi_{E}|$, $|\psi_{E}\rangle =\sum_{j}b_j|j\rangle $, 
and
$|\psi_I\rangle =\sum_{x}a_x|x\rangle $, we have \cite{KT}
\begin{align}
& \Tr_I |\psi_I\rangle \langle \psi_I| \Tr_{E} U (I\otimes \rho_{E}) U 
=
\sum_{x,j} | \langle \psi_I, j| U |x, \psi_{E}\rangle|^2\nonumber\\
=&
\sum_{x,j'} \Big|
\sum_y \bar{a_y} \sum_j b_j \langle y, j'| U |x, j\rangle \Big|^2 \nonumber\\
=&
\sum_{x,j'} \Big|
\sum_y \bar{a_y} \sum_j b_j \langle y, -j| U |x, -j'\rangle \Big|^2 \nonumber\\
=&
\sum_y \bar{a_y} \sum_j b_j 
\sum_{\tilde{y}} \bar{a_{\tilde{y}}} \sum_{\tilde{j}} b_{\tilde{j}} 
\langle y, -j| y, -\tilde{j} \rangle  \nonumber\\
=&
\sum_y |\bar{a_y}|^2 \sum_j |b_j|^2 
=1. \Label{10-24-3}
\end{align}
Hence, when 
$\rho_{E}=\sum_l p_l |\psi_{E,l}\rangle \langle \psi_{E,l}|$, 
we also have
\begin{align}
\Tr_I |\psi_I\rangle \langle \psi_I| \Tr_{E} U (I\otimes \rho_{E}) U 
=1.
\end{align}

Next, we proceed to the proof of the general case.
The above discussion implies that
\begin{align}
\Tr_{E1} U 
(\frac{\hat{1}_I}{d_I} \otimes \frac{\hat{1}_{E2}}{d_{E2}}
\otimes \rho_{E1} )U^\dagger =
\frac{\hat{1}_I}{d_I}\otimes \frac{\hat{1}_{E2}}{d_{E2}}.
\end{align}
The information processing inequality yields that
\begin{align}
&D\Big(
\frac{\hat{1}_I}{d_I}\otimes \rho_{E2}
\Big\|
\frac{\hat{1}_I}{d_I}\otimes \frac{\hat{1}_{E2}}{d_{E2}} \Big)\nonumber\\
\ge &
D\Big(
\Tr_{E1} U 
(\frac{\hat{1}_I}{d_I}\otimes  \rho_{E2}
\otimes \rho_{E1} )U^\dagger 
\Big\|\nonumber \\
& 
\Tr_{E1} U 
(\frac{\hat{1}_I}{d_I}\otimes \frac{\hat{1}_{E2}}{d_{E2}}
\otimes \rho_{E1} )U^\dagger 
\Big) \nonumber\\
\ge &D\Big(
\Tr_{E1} U 
(\frac{\hat{1}_I}{d_I}\otimes  \rho_{E2}
\otimes \rho_{E1} )U^\dagger 
\Big\| \frac{\hat{1}_I}{d_I}\otimes \frac{\hat{1}_{E2}}{d_{E2}}
\Big),
\end{align}
which implies that
\begin{align}
S(\Tr_{E1} U 
(\frac{\hat{1}_I}{d_I}\otimes  \rho_{E2}
\otimes \rho_{E1} )U^\dagger )
\ge
\log d_I+ S(\rho_{E2}).\Label{3-14-9eq}
\end{align}
Due to the condition \eqref{3-14-12eq},
the reduced density operator $\Tr_{E1} U 
(\frac{\hat{1}_I}{d_I}\otimes  \rho_{E2}
\otimes \rho_{E1} )U^\dagger $ on $E2$
is $\rho_{E2}$.
Under this condition, we have the inequality \eqref{3-14-9eq}.
Equality in \eqref{3-14-9eq} holds only when
$\Tr_{E1} U 
(\frac{\hat{1}_I}{d_I}\otimes  \rho_{E2}
\otimes \rho_{E1} )U^\dagger =
\frac{\hat{1}_I}{d_I}\otimes \rho_{E2}$,
which implies the unital condition \eqref{3-14-10eq}.
\end{proof}

Overall, as a realizable heat engine, 
we impose 
the energy-conserving, shift-invariant, and 
stationary conditions to our FQ-work extractions.
Also, it is natural to assume that 
the initial state on $E1$ is an eigenstate of the Hamiltonian 
$\hat{H}_{E1}$
because it is not easy to prepare a non-eigenstate 
of the Hamiltonian $\hat{H}_{E1}$.
{Namely, we define the standard FQ-work extraction as follows:}
\begin{definition}
When a FQ-work extraction satisfies {the energy-conserving, shift-invariant, 
stationary conditions, and the condistion that the initial state on $E1$ is an eigenstate of the Hamiltonian 
$\hat{H}_{E1}$}, we refer to it as a {\it standard FQ-work extraction}.
\end{definition}

So, for a standard FQ-work extraction ${\cal F}$,
the CP-work extraction $CP({\cal F})$ is a standard CP-work extraction.
In the following, 
we consider that the set of standard FQ-work extractions
as the set of preferable work extractions.
Hence, our optimization will be done among the set of standard FQ-work extractions.

However, we can consider restricted classes of standard FQ-work extractions
by considering additional properties.
For a standard FQ-work extraction ${\cal F}=
({\cal H}_{E}, \hat{H}_{E}, U, \rho_{E})$, 
we assume that
the external system $\cH_{E}$ consists only of  
a non-degenerate external system $\cH_{E1}$.
In this case, the corresponding standard CP-work extraction
$CP({\cal F})$ depends only on the internal unitary 
$U_I= WUW^\dagger$.

\begin{definition}
For the above given standard FQ-work extraction ${\cal F}=
({\cal H}_{E}, \hat{H}_{E}, U, \rho_{E})$,
the standard CP-work extraction
$CP({\cal F})$ is called
{\it the standard CP-work extraction associated to 
the internal unitary $U_I$},
and is denoted by $\hat{CP}(U_I)$.
\end{definition}

Further, 
an internal unitary $U_I$
is called {\it deterministic}
when 
$\langle x| U_I |x'\rangle$
is zero or $e^{i\theta}$ for any $x$ and $x'$.
In the latter, 
a standard CP-work extraction associated to 
the deterministic internal unitary 
plays an important role.

We illustrate a Venn diagram of the FQ-work extractions in Figure \ref{venn2}, and classify the relationships among the classes of the CP-work extraction and that of the FQ-work extraction in Table \ref{table}.
\begin{figure}[htbp]
\begin{center}
\includegraphics[scale=0.9]{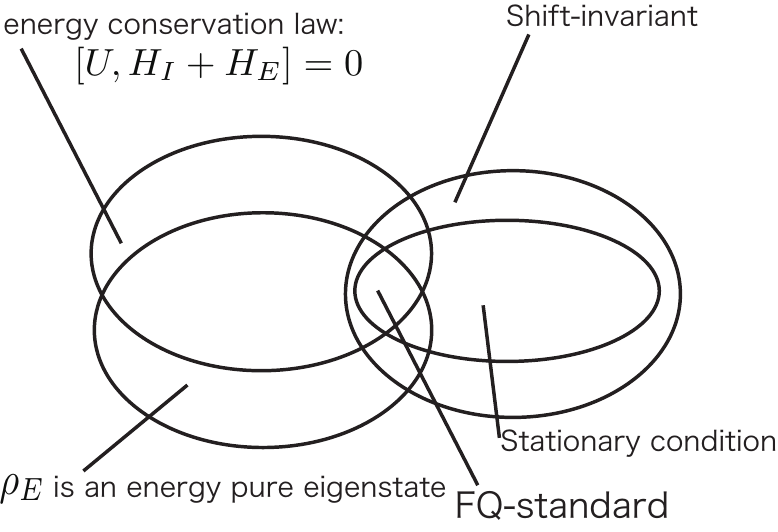}
\end{center}
\caption{A Venn diagram of the FQ-work extractions}
\Label{venn2}
\end{figure}

\begin{table*}[htb]
\caption{Correspondence between FQ-work extraction and CP-work extraction.
The following table describes 
the type of CP-work extraction that can be generated by respective classes of FQ-work extraction.}
  \begin{tabular}{|c|c|c|c|c|}\hline
\multirow{4}{*}{$CP({\cal F})$} & 
FQ-energy & 
FQ-energy & 
\multirow{2}{*}{Stationary} & 
FQ-energy 
\\
&conservation, & conservation, & & conservation, 
\\
&$\rho_{E}$ is not an & 
\multirow{2}{*}{Shift-invariant} &
\multirow{2}{*}{condition} & 
$\rho_{E}$ is an 
\\
& energy eigenstate  & 
& & energy eigenstate 
\\
\hline
CP-level 2 & Yes (Lemma \ref{4-2-1L}) &Yes & - & Yes 
\\\hline
CP-level 3 & - & Yes (Lemma \ref{3-13-2L}) & - & Yes 
\\\hline
CP-level 4 & No (Theorem \ref{3-14-8L})& - & - & Yes (Theorem \ref{3-14-8L})  
\\\hline
CP-unital  & -& - & Yes (Theorem \ref{3-13-5L}) & -
\\\hline
  \end{tabular}
\Label{table}
\end{table*}


Finally, we consider the relation with the  the classical standard formulation.
In the shift-invariant unitary,
if we focus on the eigenstates,
the state $|i,j\rangle $ can be regarded as being probabilistically changed to
$|i',j'\rangle $ with the energy conserving law 
\begin{align}
h_i+h_Ej =h_{i'}+h_E j' , \Label{7-14-1}
\end{align}
which is the same as in  the classical standard formulation.
When we focus on the internal system,
$|i\rangle $ is changed to $|i'\rangle $.
So, in the semi-classical scenario, they make a unitary time evolution based on the states on the internal system.
However, the time evolution occurs between the internal and external systems as
$|i,j\rangle \mapsto |i',j'\rangle $ in the classical standard formulation
under the condition \eqref{7-14-1}.
So, to consider its natural quantum extension, we need to make a unitary evolution on the composite system.
That is, it is natural to add proper complex amplitudes to the time evolution $|i,j\rangle \mapsto |i',j'\rangle $ so that it forms a unitary evolution on the composite system.
Hence, our shift-invariant unitary can be regarded as a natural quantum extension of  the classical standard formulation because it is constructed in this way.

One might consider that there is 
a cost for initialization of the measurement.
However, if we adopt the projection postulate, this problem can be resolved when we employ a shift-invariant model.
The classical standard formulation does not describe the dynamics of the external system, explicitly
because the dynamics of the internal system does not depend on the state of the external system. 
Similarly, under our shift-invariant model,  
the dynamics of the internal system does not depend on the state of the external system 
as long as the initial state of the external system is an energy eigenstate.
When the projection postulate holds for our measurement, the final state of the external system is an energy eigenstate,
which can be used as the initial state for the next work extraction. 
In real systems, decoherence occurs during unitary evolution, such that
the state of the external system becomes inevitablly an energy eigenstate.

Further, we should notice that the initialization of the measurement 
is different from the initialization of the thermal bath.
In the finite-size setting, the final state in the thermal bath
is different from the thermal state in general.
Hence, we need to consider the initialization of the thermal bath.
However, the thermal bath is considered as a part of the internal system in our model.
So, this problem is different from the initialization of the measurement,
and will be discussed in \cite{pre2}.

Although we consider only the lattice case,
in real systems, there are so many cases with non-lattice Hamiltonian.
Our discussion can be extended to the non-lattice case with small modification. The detail is given in Appendix \ref{s4-3}.

\section{Trade-off between information loss and coherence}\Label{s5}
\subsection{Approximation and coherence}\Label{s5a}
In this section, we investigate the validity of the semi-classical scenario \cite{tasaki,Croocks,Car1,sagawa1,jacobs,sagawa2,Funo,Morikuni,Parrondo,IE1,IE2,IE1.5,BBM1,TLH}
 based on an 
FQ-work extraction model.
To discuss this issue,
we firstly recall the semi-classical scenario as follows.
\begin{description}
\item[\bf Semi-classical scenario \cite{tasaki,Croocks,Car1,sagawa1,jacobs,sagawa2,Funo,Morikuni,Parrondo,IE1,IE2,IE1.5,BBM1,TLH}:]
In this scenario, we also consider that
to extract work, an external agent performs the external operation as the unitary time evolution $U_{I}:={\cal T}[\exp(\int-i\hat{H}_{I}(t)dt)]$  (${\cal T}[...]$ denotes time-ordered product) on $I$ by time-dependently varying the control parameter of the Hamiltonian $\hat{H}_{I}(t)$ of the internal system $I$, which usually consists of the system $S$ and the heat bath $B$. 
During the time evolution, the loss of energy of the internal system
is transmitted to the external controller through the back reaction of the control parameter. So, the energy loss is regarded as the extracted work\cite[Section 2]{tasaki}.
This scenario is considered as a natural quantum extension of the classical standard formulation \cite{Ehrenfest}.
\end{description}

In the semi-classical scenario, they expect that the unitary can be realized via the control of the Hamiltonian of the internal system.
Since the control is realized by the external system,
they consider that
the work can be transferred to the external system via the control.
In this scenario, they expect that
the time evolution $\Lambda$ of the internal system $\cH_I$
can be approximated to an ideal unitary $U_I$.
That is, when we employ an FQ-work extraction
${\cal F}=({\cal H}_{E}, \hat{H}_{E}, U, \rho_{E})$,
the time evolution $\Lambda$ of the internal system $\cH_I$
is given as
$\Lambda(\rho_I)= \Tr_{E} U (\rho_I \otimes \rho_{E}) U^\dagger$.
To qualify the approximation to the unitary $U_I$, we need to focus on two aspects.
One is the time evolution of basis states in a basis 
$\{|x\rangle \}_{x}$ diagonalizing the Hamiltonian $\hat{H}_I$.
The other is the time evolution of superpositions of states in this basis.
Usually, it is not difficult to realize the same evolution as that of $U_I$ only for the former states. 
However, it is not easy to keep the quality of the latter time evolution, which is often called the coherence.
Hence, we fix a unitary $U_I$, and we assume that
the TP-CP map $\Lambda$ satisfies the condition
\begin{align}
\langle y|\Lambda(|x \rangle \langle x|)|y\rangle = 
\langle y| U_I |x \rangle \langle x|U_I^\dagger |y\rangle 
\hbox{ for any }x,y.\Label{3-8-1}
\end{align}
That is, we choose our time evolution among TP-CP maps satisfying the above condition.
Under the condition, 
the quality of the approximation to the unitary $U_I$ 
can be measured by the coherence of the TP-CP map $\Lambda$.
For a measure of coherence, we employ the entropy exchange \cite{Schumacher}.
\begin{align}
S_e(\Lambda,\rho_I):=
S(\Lambda(|\Phi \rangle \langle \Phi |)),\Label{3-8-12}
\end{align}
{where $|\Phi \rangle$ is the purification of the state $\rho_I$ with the reference system $R$.
When a TP-CP map $\Lambda$ satisfies \eqref{3-8-1} and $S_e(\Lambda,\rho_I)<<1$, we say that ``$\Lambda\approx U_{I}$.''}

\subsection{Expression of detectability of work in terms of correlation}\Label{s5b}

{On the other hand, we need to focus on another feature of thermodynamics, i.e., the detectability of work. (Fig.\ref{detectwork})
In standard thermodynamics, we can measure the amount of extracted work by measuring only on the work storage system, e.g., the weight.
In order that the heat engine works properly as energy transfer from a collection of microscopic systems to a macroscopic system,
the amount of extracted work needs to be reflected to the outer system that consists of the macroscopic devices.
Hence, the amount of energy loss in the internal system
is required to be correlated to the outer system.
So, we can measure the detectability of work by the correlation between 
the external system $E$ after energy extraction 
and the energy loss in the internal system $I$. 
In Fig. \ref{tradeofffig1},
the former is expressed as A and the latter is as B.
Our purpose of this section is to show a trade-off relation, which indicates 
the impossibility of the detection of work under the condition $\Lambda\approx U_{I}$.
Namely, when $\Lambda\approx U_{I}$, the correlation between A and B in Fig. \ref{tradeofffig1} is close to zero.}

\begin{figure}[htbp]
\begin{center}
\includegraphics[scale=0.35]{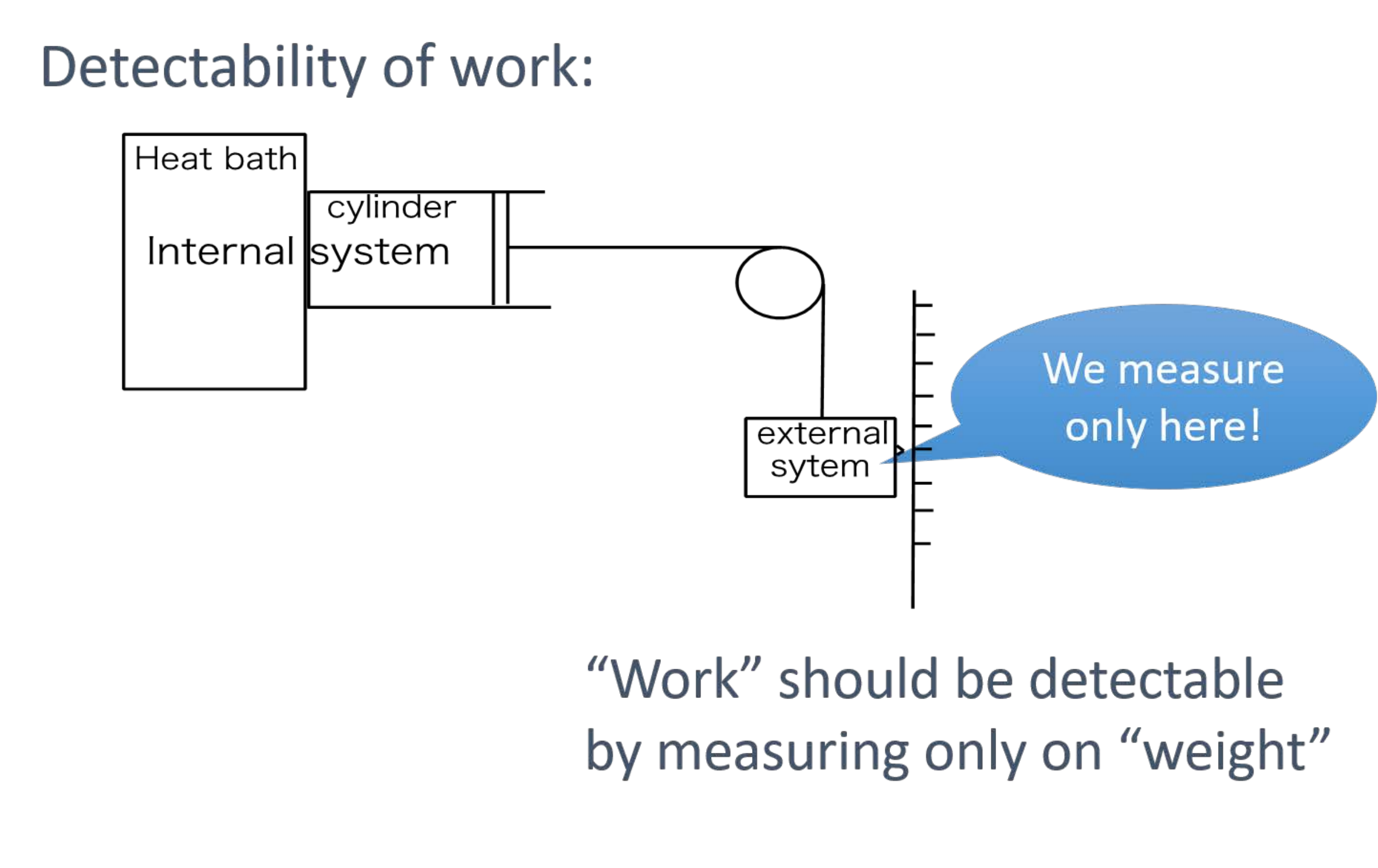}
\end{center}
\caption{The concepts of the detectability of work.
}
\Label{detectwork}
\end{figure}

\begin{figure}[htbp]
\begin{center}
\includegraphics[scale=0.35]{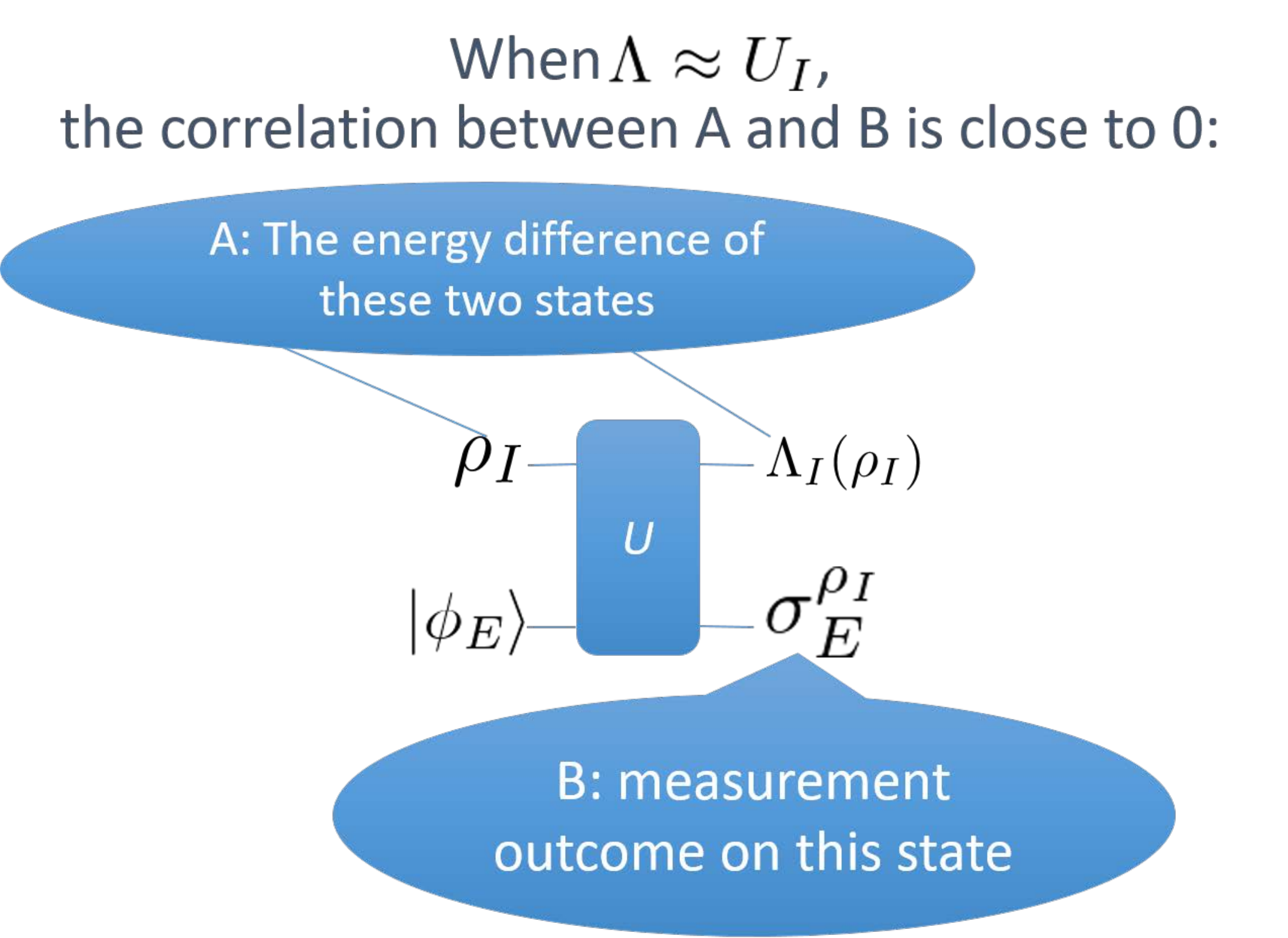}
\end{center}
\caption{The concept of our trade-off relation.
Here, $\sigma^{\rho_{I}}_{E}:=\Tr_{I}[U(\rho_{I}\otimes|\phi\rangle\langle\phi|_{E})U^{\dagger}]$ }
\Label{tradeofffig1}
\end{figure}

{
In order to show the trade-off relation, we employ the purification 
$|\Phi\rangle$ of the state $\rho_I$ with the reference system $R$,
which satisfies $\Tr_{I}[|\Phi\rangle\langle\Phi|]=\Tr_{E}[|\Phi\rangle\langle\Phi|]=\rho_{I}$.
Then, we can interpret $R$ as a kind of memory, and can translate the energy loss of $I$ during the time evolution $\Lambda$ into the energy difference between $I$ and $R$ after the time evolution. (Fig.\ref{detect3})
We also employ the initial state $\rho_{E}$ on the external system $E$. 
Here, we take the external system large that
the time evolution on the joint system of $I$ and $E$ can be regarded as a unitary $U$ on $IE$
and the state $\rho_{E}$ is a pure state $|\phi_{E}\rangle$.
}
In the following, we denote the initial state 
$|\Phi,\phi_{E}\rangle\langle \Phi,\phi_{E}|$
of the total system $\cH_I\otimes \cH_R \otimes \cH_{E}$
by $\rho_{IRE}$.
Then, we denote the resultant state 
$U|\Phi,\phi_{E}\rangle\langle \Phi,\phi_{E}|U^\dagger$
by $\rho'_{IRE}$. 
\begin{figure}[htbp]
\begin{center}
\includegraphics[scale=0.32]{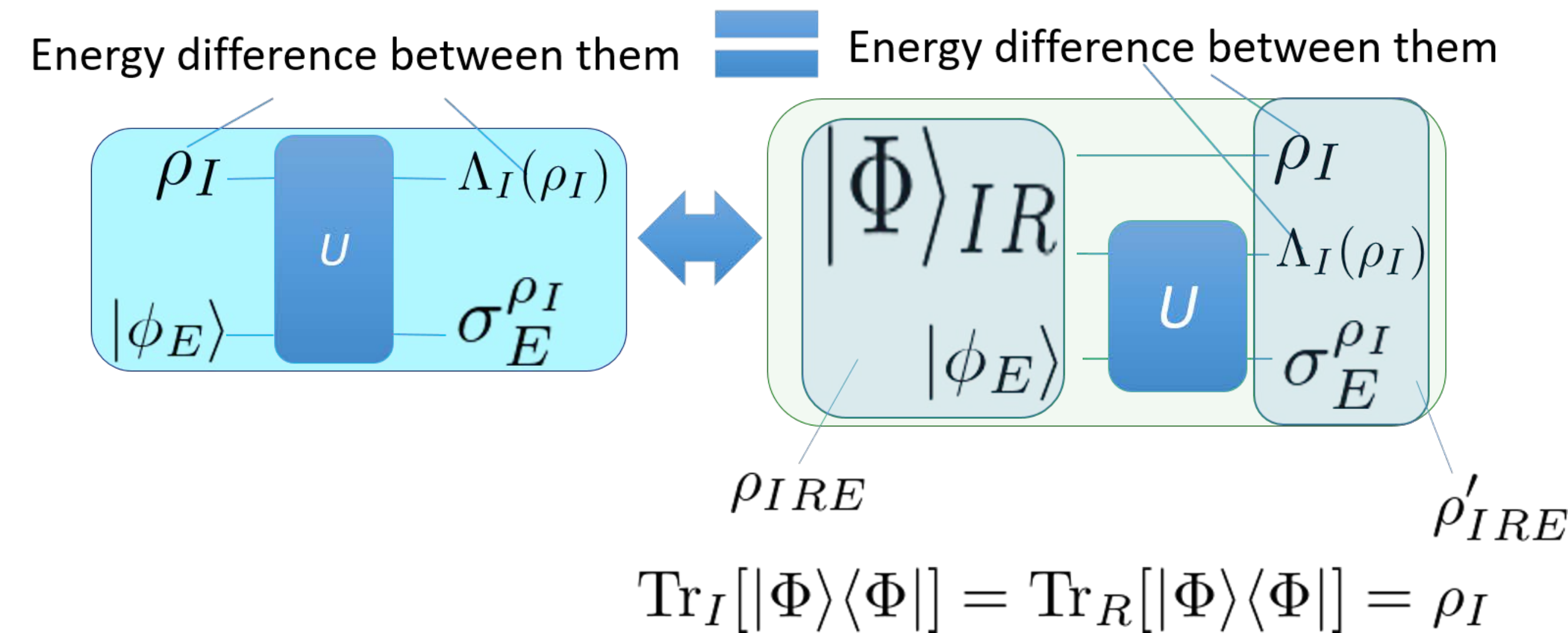}
\end{center}
\caption{We can translate the energy loss of $I$ during the time evolution $\Lambda$ into the energy difference between $I$ and $R$ }
\Label{detect3}
\end{figure}

To determine the amount of energy lost from the internal system, we consider a measurement of
the observable $\hat{H}_I-\hat{H}_R$ on the joint system $I$ and $R$ after the time evolution $\Lambda$.(Fig. \ref{detect4})
That is, by using the spectral decomposition $\hat{H}_I-\hat{H}_R=\sum_{z}z F_z$,
we define the final state $\rho''_{ZE}$ on $\cH_Z\otimes \cH_{E} $ by
\begin{align}
\rho''_{ZE}:=
\sum_{z}|z\rangle_Z~_Z \langle z| \otimes \Tr_{IR} F_z \rho'_{IRE}. 
\end{align}
Here, the random variable $Z$ takes the value $z$ with probability
$P_Z(z):= \Tr \rho'_{IRE} F_z
=\Tr \Lambda(|\Phi \rangle \langle \Phi|) F_z$.
When the outcome $Z$ is $z$, the resultant state on $E$ is 
$\rho_{E|Z=z}''
:=\frac{1}{P_Z(z)} \Tr_{IR} \rho'_{IRE} F_z $.
So, the correlation between the external system $E$ and the amount of energy loss
can be measured by the correlation on the state $\rho_{Z,E}$.
Although this scenario is based on a virtual measurement,
this interpretation will be justified by 
considering the situation without the purification in 
the following lemma.

\begin{figure}[htbp]
\begin{center}
\includegraphics[scale=0.3]{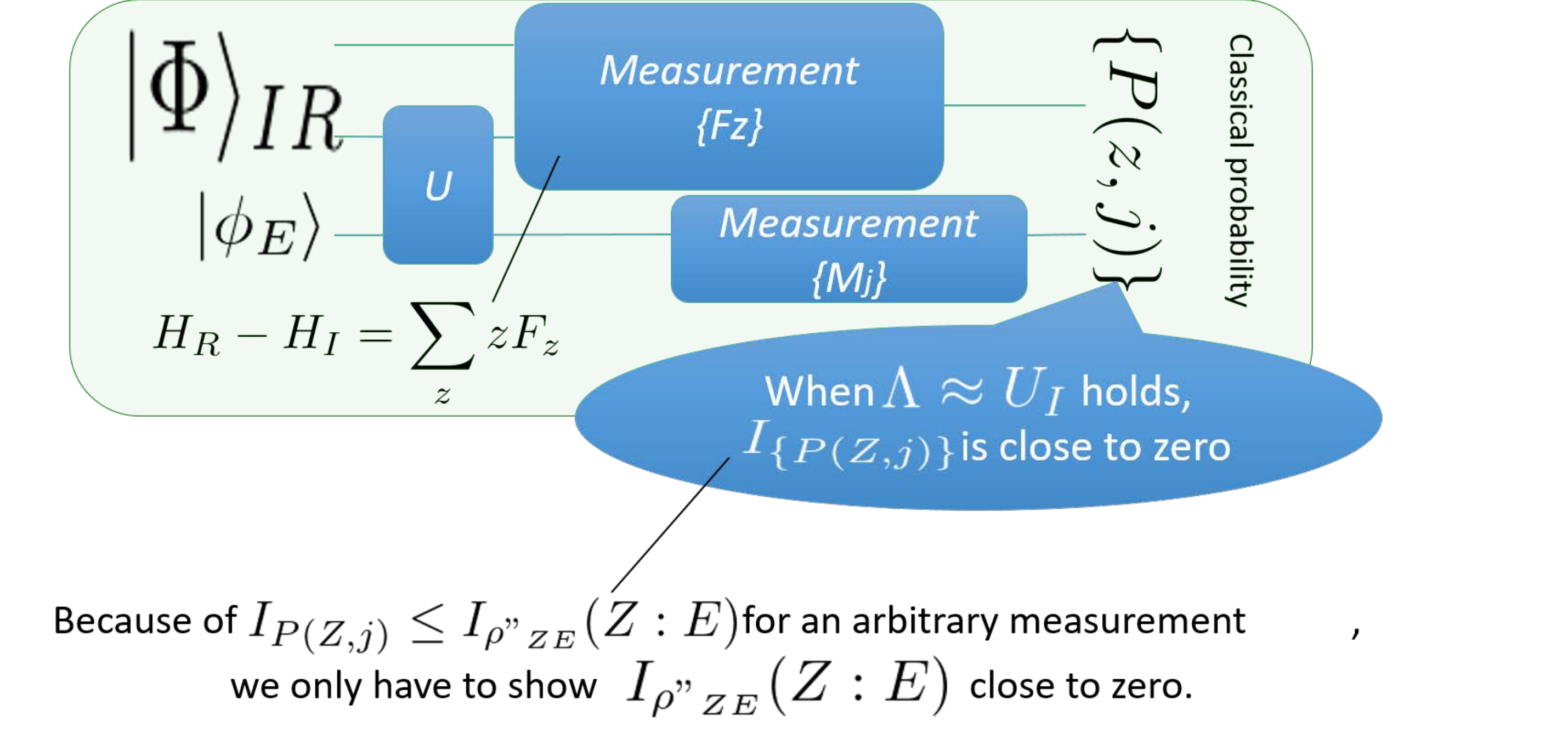}
\end{center}
\caption{The schematic diagram of the formulation.}
\Label{detect4}
\end{figure}

\begin{lemma}\Label{3-15-3L}
We consider the case when 
the eigenstate $|\psi_{I,a} \rangle$ is generated 
with probability $P_A(a)$.
Then, measuring the Hamiltonian $\hat{H}_I=\sum_h h P_h$ after the time evolution $\Lambda$, 
we obtain the conditional distribution $
P_{H|A}(h|a):=\Tr P_h 
\Lambda(|\psi_{I,a} \rangle \langle \psi_{I,a} |)$. 
Then, we chose the state $\rho_I$ to be
\begin{align}
\rho_I=\sum_a P_A(a) |\psi_{I,a} \rangle \langle \psi_{I,a}| .
\Label{3-15-1eq}
\end{align}
Under the joint distribution 
$P_{HA}(h,a):=P_{H|A}(h|a)P_A(a) $,
the amount $\langle \psi_{I,a} | \hat{H}_I |\psi_{I,a} \rangle -h$ of energy lost takes values 
$z$ with the probability $P_{Z}(z)$.
That is,
\begin{align}
\sum_{h,a:\langle \psi_{I,a} | \hat{H}_I |\psi_{I,a} \rangle-h=z }P_{HA}(h,a) =P_{Z}(z).
\end{align}
\end{lemma}

Since the proof of Lemma \ref{3-15-3L} is trivial, we skip its proof.
Lemma \ref{3-15-3L} guarantees that the probability distribution $P_{Z}(z)$ is the same as the work distribution defined in \cite{TLH} under its assumption.
When $\rho_I$ commutes with $\hat{H}_I$, we can take the above 
eigenstate $|\psi_{I,a} \rangle$ and a probability distribution $P_A$
satisfying the above condition \eqref{3-15-1eq}.
Also, when the state $\rho_I$ is a Gibbs state or a mixture of Gibbs states,
the state $\rho_I$ satisfies the condition in Lemma \ref{3-15-3L}.
Hence, it is suitable to consider the distribution $P_{Z}$ of the amount of energy lost in this case.

Here, we employ two measures of the correlation.
One of the measures is the mutual information:
\begin{align}
& I_{\rho''_{ZE}}(Z;E) 
:= D(\rho''_{ZE} \| \rho_{Z}'' \otimes \rho_{E}'' )\nonumber\\
=&
S(\rho_{E}'' )- \sum_{z} P_Z(z) S(\rho_{E|Z=z}'' )
\stackrel{(a)}{\ge } S(\rho_{E}'' )
=S(\Lambda(|\Phi \rangle \langle \Phi |)),
\Label{3-8-11}
\end{align}
where $D(\tau\|\sigma):= \Tr \tau (\log \tau -\log \sigma)$.
Here, equality in $(a)$ holds 
if and only if the state $\rho_{E|Z=z}''$ 
is pure for any value $z$ with non-zero probability $P_Z(z)$.
The information processing inequality 
for the map $|z\rangle \langle z| \mapsto \rho_{E|Z=z}''$ yields
\begin{align}
I_{\rho''_{ZE}}(Z;E) \le S(P_Z).
\Label{3-11-1eq}
\end{align}
In particular, equality in \eqref{3-11-1eq} holds
in the ideal case, i.e., in the case when 
the states $\{ \rho_{E|Z=z}''\}_z$
are distinguishable, i.e.,
$\Tr \rho_{E|Z=z}'' \rho_{E|Z=z'}''=0$ for $z\neq z'$
with non-zero probabilities $P_Z(z)$ and $P_Z(z')$.
For example, when $U$ satisfies the energy conservation law \eqref{FQ2},
the initial state of the internal system commutes with the internal Hamiltonian $\hat{H}_I$,
and 
the initial state of external system is an energy eigenstate,
this conditions holds because the energy of the final state of external system precisely reflects the loss of energy of the internal system.
So, the imperfectness of the correlation for the decrease of the energy can be measured by the difference
\begin{align}
\Delta I_{\rho''_{ZE}}(Z;E) 
:= S(P_Z)- I_{\rho''_{ZE}}(Z;E) .
\Label{3-8-13}
\end{align}

\subsection{Trade-off with imperfectness of correlation}\Label{s5c}
Many papers studied trade-off relations between 
the approximation of a pure state on the bipartite system 
and the correlation with the third party $E$.
In particular, since this kind of relation plays an important role in the security analysis in quantum key distribution (QKD),
it has been studied mainly from several researchers in QKD and the related areas with various formulations 
\cite[Section V-C]{Schumacher}
\cite[(21)]{Hamada}
\cite[Theorem 1]{Koashi}\cite[Lemma 2]{Hayashi1}\cite[Theorem 2]{Hayashi2}\cite{Miyadera}\cite{Buscemi2008}.

However, these trade-off relations are not suitable for our situation.
So, we derive two kinds of trade-off relations 
with the imperfectness of correlation, which are more suitable for our purpose.
\begin{theorem}\Label{3-15-2T}
The amount of decoherence $S_e(\Lambda,\rho_I)$
and 
the amount of imperfectness of correlation $\Delta {I}_{\rho''_{ZE}}(Z;E) $
satisfy the following trade-off relation:
\begin{align}
S_e(\Lambda,\rho_I)
+\Delta {I}_{\rho_{ZE}''}(Z;E) 
&\ge S(P_Z) \Label{3-8-9}.
\end{align}
The equality holds 
if and only if the state $\rho_{E|Z=z}''$ 
is a pure state $| \psi_{E|Z=z}\rangle$
for any value $z$ with non-zero probability $P_Z(z)$.
In this case, we have
\begin{align}
S_e(\Lambda,\rho_I)
&= {I}_{\rho''_{ZE}}(Z;E) 
= S\Big(\sum_{z}P_Z(z) 
|\psi_{E|Z=z}\rangle \langle \psi_{E|Z=z}|
\Big)
\Label{3-11-7}.
\end{align}
\end{theorem}

\begin{proof}
Since \eqref{3-8-11}, \eqref{3-8-12} and \eqref{3-8-13} imply
\begin{align}
&S_e(\Lambda,\rho_I)
+\Delta {I}_{\rho''_{ZE}}(Z;E) 
=S(P_Z)- I_{\rho''_{ZE}}(Z;E)+S(\rho_{E}'')\nonumber\\
&=S(P_Z)+ \sum_{z} P_Z(z) S(\rho_{E|Z=z}'' )
\ge S(P_Z).
\end{align}
Equality holds if and only if
the state $\rho_{E|Z=z}''$ 
is pure for any value $z$ with non-zero probability $P_Z(z)$.
In this case, we have $\rho_{E}''
=\sum_{z}P_Z(z) 
|\psi_{E|Z=z}\rangle \langle \psi_{E|Z=z}|$,
which implies \eqref{3-11-7}.
\end{proof}

Due to the relation \eqref{3-8-9},
when the imperfectness $\Delta {I}_{\rho''_{ZE}}(Z;E)$
of the correlation is close to zero,
the amount of decoherence $S_e(\Lambda,\rho_I)$ is far from zero.
So, the coherence cannot be kept in this work extraction process.
That is, the unitary $U_I$ cannot be approximated by the actual time evolution $\Lambda$.
On the other hand, when the amount of decoherence
$S_e(\Lambda,\rho_I)$ is close to zero, 
the imperfectness $\Delta {I}_{\rho''_{ZE}}(Z;E)$ 
of the correlation is far from zero.
In this case, the perfect approximation is realized
although the external system has almost no correlation.

{In the current discussion, we deal with a general framework for the ensemble of initial states of the internal system.
In practice, it is quite difficult to realize 
an arbitrary distribution of the initial state.
However, without requiring to realize an arbitrary distribution,
we have a realistic  example to derive a contradiction for the semi-classical model.
That is, in such a desired example,
if the coherence is kept, which is the requirement of the semi-classical model, we cannot detect the amount of extracted work precisely.
To give such an example, it is sufficient to consider an ensemble whose entropy is greater than $\log 2$.
For example, 
a useful internal state (in which, work extraction is possible) occurs with probability $\frac{1}{2}$, 
and a useless internal state (in which, work extraction is impossible) occurs with probability $\frac{1}{2}$, i.e.,
with probability $\frac{1}{2}$, we cannot extract work from the initial state.
In this example, when the amount of decoherence
$S_e(\Lambda,\rho_I)$ is close to zero, 
the imperfectness $\Delta {I}_{\rho''_{ZE}}(Z;E)$ 
of the detection of the amount of extracted work 
is close to $\log 2$, which is far from zero.
Since such an ensemble is possible in the real world, 
this example shows that the semi-classical model is not suitable for the model of a heat engine
because we cannot detect  the amount of extracted work, precisely.}

\subsection{Shift-invariant model}\Label{s5d}
Next, we consider how to realize the case 
when the amount of the decoherence $S_e(\Lambda,\rho_I)$ is close to zero.
The following theorem is the solution to this problem, which will be shown in Appendix \ref{s5d-s}.
\begin{theorem}\Label{th11-23}
We assume the shift-invariant model without an additional external system.
When
\begin{align}
|\psi_{E}\rangle&=\sum_j \sqrt{P_J(j)}|j\rangle \\
P_J(j)&=
\left\{
\begin{array}{ll}
\frac{1}{2m+1} & \hbox{ if } |j| \le m \\
0 & \hbox{ otherwise.}
\end{array}
\right. \\
P_Z(h_E j)& =0  \hbox{ if } |j| \ge l,
\end{align}
we have
\begin{align}
 S_e(\Lambda,\rho_I) 
\le & 
h\big(2\frac{l}{2m+1}- (\frac{l}{2m+1})^2\big)\nonumber \\
&+\big(2\frac{l}{2m+1}- (\frac{l}{2m+1})^2\big)
\log (d_I^2-1).
\end{align}
\end{theorem}
So, when $m$ is sufficiently large,
the quality of approximation is very small under 
the shift-invariant model without an additional external system.
That is, a large superposition enables us to keep coherence.
This consequence is qualitatively consistent with the conclusion in \cite[Proposition 2 of Supplement]{catalyst}, 
which assess the quality of approximation with the trace norm
unlike Theorem \ref{th11-23}.

\subsection{Trade-off under CP-work extraction}\Label{s5e}
In Subsection \ref{s5c}, we discussed the trade-off relation between the 
imperfectness of correlation and the quality of approximation under a FQ-work extraction.
In this subsection, we discuss the trade-off relation under a CP-work extraction 
${\cal G}:=\{{\cal E}_j,w_j\}_j$ on the internal system ${\cal H}_I$ with Hamiltonian $\hat{H}_I$.
To discuss the trade-off relation, in the internal system ${\cal H}_I$,
we consider the internal unitary $U_I$ to be approximated and
the initial mixted state $\rho_I$, which is assumed to commute with $\hat{H}_I$.
To quantify the approximation, we employ the measure
$S_e(\sum_j {\cal E}_j,\rho_I)$.
To evaluate the imperfectness of correlation, 
we consider the purification $|\Phi\rangle$ of $\rho_I$,
and introduce the joint distribution $P_{ZW}$ as
\begin{align}
P_{ZW}(z,w):=\sum_{j:w_j=w} \Tr_{R I} {\cal E}_j(|\Phi\rangle \langle \Phi|) F_z,
\end{align}
where the projection $F_z$ is defined in the same way as in Subsection \ref{s5b}.
Then, we employ the measure
$\Delta I_{P_{ZW}}(Z;W) :=S(P_Z)-I_{P_{ZW}}(Z;W)$.

Since the measurement outcome precisely reflects decrease in energy of the internal system
in a level-4 energy CP-work extraction $\{{\cal E}_j,w_j\}_j$, 
we have the following lemma.

\begin{lemma}\Label{L5-11}
For a level-4 energy CP-work extraction $\{{\cal E}_j,w_j\}_j$, 
the relation $\Delta I_{P_{ZE}}(Z;E)=0$ holds for any internal state $\rho_I$ that commutes with $\hat{H}_I$.
\end{lemma}

This lemmas shows that a level-4 energy CP-work extraction satisfies our requirement explained in introduction.
As a corollary of Theorem \ref{3-15-2T}, we have the following by virtually considering the indirect measurement.
\begin{corollary} \Label{4-10-1C}
Given a CP-work extraction $\{{\cal E}_j,w_j\}_j$.
The amount of decoherence $S_e(\sum_j {\cal E}_j,\rho_I)$
and the amount of imperfectness of correlation $\Delta {I}_{P_{ZE}}(Z;E) $
satisfy the following trade-off relation:
\begin{align}
S_e(\sum_j {\cal E}_j,\rho_I)
+\Delta I_{P_{ZW}}(Z;W) 
&\ge S(P_Z) \Label{3-8-9x}.
\end{align}
\end{corollary} 

\begin{proofof}{Corollary \ref{4-10-1C}}
Notice that Theorem \ref{3-15-2T} does not assume any energy conservation law.
Then, we take a Stinespring representation 
$({\cal H}_E, U, \rho_E)$ with a pure state $\rho_E$ of 
$\{{\cal E}_{j},w_{j}\}_{j\in {\cal J}}$ as follows.
\begin{align}
{\cal E}_{j}(\rho_I)=
\Tr_{E} U_{IE}
(\rho_I \otimes \rho_{E}) U_{IE}^\dagger 
(\hat{1}_{I} \otimes E_j),
\end{align}
where $\{E_j\}$ is a projection valued measure on ${\cal H}_E$.
Notice that $({\cal H}_E, U, \rho_E)$ is an FQ-work extraction.
Here, we do not care about whether the FQ-work extraction satisfies any energy conservation law.
Then, we apply Theorem \ref{3-15-2T}.
Since the information processing inequality for the relative entropy yields
$I_{\rho''_{ZE}}(Z;E) \ge I_{P_{ZE}}(Z;E)$,
the relation \eqref{3-8-9} implies
\eqref{3-8-9x}.
\end{proofof}

Further, we have the following corollary.
\begin{corollary}
When $\Delta I_{P_{ZE}}(Z;E)=0$, we have
\begin{align}
S_e\Big(\sum_j {\cal E}_j,\rho_I\Big)
&\ge S(P_Z) \Label{3-8-9z}.
\end{align}
\end{corollary}
The combination of this corollary and Lemma \ref{L5-11}
says that the dynamics of a level-4 CP-work extraction is far from any internal unitary.

\section{Remark: other consequences of our formulation and relation to other formulation}
Finally, we discuss the relation to other formulations.
While we have introduced four kinds of energy conservation laws for measurement-based work extraction,
a level-1 CP-work extraction can be mathematically regarded as an average work extraction
\cite{Car2} as follows.
In \cite[Section II-A-2]{Optimal}, we formulated the average work extraction 
by using the TP-CP map on the internal system in an implicit treatment of the external work storage. 
For a given level-1 CP-work extraction $\{{\cal E}_j,w_j \}_{j\in {\cal J}}$, 
we define a TP-CP map $\sum_j {\cal E}_j$, which gives an average work extraction with implicit treatment of the external work storage. 
Hence, any level-1 CP-work extraction can be treated as an average work extraction.
Thus, any model in our paper is mathematically a part of average work extraction
with implicit treatment of the external work storage.
Since the optimal efficiency of an average work extraction 
asymptotically equals the Carnot efficiency \cite{Popescu2014,Optimal,Popescu2015},
the efficiency of any our model does not exceed the Carnot efficiency.

Now, we discuss the detailed relation with existing models.
While our measurement-based model for work extraction
treats energy transfer from a collection of microscopic systems to
a macroscopic system,
our model can treat energy transfer in the microscopic scale as follows.
As shown in Lemma \ref{3-13-4L},
a level-4 CP-work extraction can be extended to
an energy-conserving and shift-invariant FQ-work extraction, which is the same as in \cite{catalyst}.
Conversely, as shown in Lemma \ref{4-2-1L}, 
an energy-conserving FQ-work extraction can be converted to a level-2 CP-work extraction.
In particular, 
the combination of Theorem \ref{3-14-8L} and Lemma \ref{3-13-4L} show that
a shift invariant model up to the second order with an initial energy-eigenstate $\rho_E$ on the external system is equivalent to a level-4 CP-work extraction.

Another paper by the same authors \cite{Optimal} derives
the higher order expansion of the optimal efficiency
under average work extraction with implicit treatment of the external work storage
while its first order equals the Carnot efficiency.
Furthermore, it was shown that the optimal efficiency can be attained by a shift invariant model
with an initial energy-eigenstate $\rho_E$ on the external system 
up to the second order.
In summary, even in any of our four models,
the optimal efficiency asymptotically has the same value up to higher order,
whereas the first order coefficient is the Carnot efficiency.
This fact shows the adequacy of our models.

As another problem, one might consider a serious increase of the entropy due to the measurement.
However, the increase is not so serious as follows.
In section 4 of Ref.\cite{Optimal}, it has been shown that we can give a concrete protocol which extract energy with negligibly small increase of entropy, while the protocol attains the optimal efficiency.
In section 4 of Ref.\cite{Optimal}, we translate the implicit formulation of an average work extraction into the explicit formulation with an external work storage by using the translation between the direct measurement and the indirect measurement.
Then, we have calculated the ratio of the energy gain to the entropy gain in the external work storage, and have shown that the entropy-energy ratio of the work storage goes to 0 in the macroscopic limit.

\section{Conclusion}
In the present article, to expand the area of quantum thermodynamics, 
we have discussed quantum heat engines as work extraction processes in terms of quantum measurement.
That is, we have formulated work extraction as a CP-instrument
when discernible energy is transfered from a collection of quantum systems to a macroscopic system.
Our formulation is so general as to include any work extraction that has an equipment to assess the amount of the extracted work
when we extract energy to a macroscopic system.
We also have clarified the relationships between the fully quantum work extraction and the CP-work extraction in our context.

Moreover, to clarify the problem in the semi-classical scenario,
we have given a trade-off relation for the coherence and information acquisition for the amount of extracted work.
The trade-off relation means that we have to demolish the coherence of the thermodynamic system in order to know the amount of the extracted energy.
Further, we have pointed out that 
our shift-invariant unitary is a natural quantum extension of the classical standard formulation rather than the semi-classical scenario in the end of Section \ref{s4b}.
In summary, these results imply the incompatibility of the coherence of the internal system
and information acquisition for the amount of extracted work.
That is, when the time evolution of the internal system is close to unitary, we cannot know the amount of the extracted work. 

In Appendix, we also show the reduction to the classical model.
When the initial state of the internal system commutes with the 
Hamiltonian,
any CP-work extraction can be simulated by a classical work extraction
in the non-asymptotic sense,
whose detailed definition is shown in Appendix.
This conversion is helpful for analyzing the performance of the work extraction process.
This property will be employed for derivation of optimal efficiency 
with the asymptotic setting in another paper \cite{Optimal}.

A reader might raise the following question for our formulation as follows.
Although this paper requires the distinguishability of 
the amount of extracted work,
it is sufficient to recover this property only within the thermodynamic limit
because the identification of the amount of extracted work
is a classical task.
Hence, we do not need to employ measurement to formulate the quantum heat engine.
However, this idea is not correct due to the following two reasons when 
we extract energy from a collection of microscopic system to a macroscopic system.

Firstly, our trade off relation between the coherence and the distinguishability holds independently of the size of the system.
So, even though we take the thermodynamic limit,
we cannot resolve this trade-off.
So, to keep the distinguishability of 
the amount of extracted work even with the thermodynamic limit,
we need to give up the description by the internal unitary.
Second, to avoid the distinguishability,
we can employ a scenario that the extracted work is autonomously stored in a quantum storage \cite{Car2,Popescu2014,oneshot2,Horodecki,oneshot1,oneshot3,Egloff,Brandao,Max2,Malabarba}.
(The word ``autonomous" is used in Ref.\cite{Max2,Malabarba}, for example.)
This scenario works well when we extract energy from a microscopic system to another microscopic system.
However, when we extract energy from a collection of microscopic system to a macroscopic system,
even in this scenario, we ultimately consume the work stored in the quantum storage.
In the consuming stage, we have to consider the distinguishability of the 
amount of work extracted from the quantum storage.
Therefore, we cannot avoid to distinguish the amount of extracted work.
That is, the solution depends on the situation, 
and the measurement is inevitably essential for a proper formulation of quantum heat engine in our situation.

\section*{Acknowledgments}
The authors are grateful to Professor Mio Murao, Professor Schin-ichi Sasa,
and Dr. Paul Skrzypczyk for helpful comments.
MH is thankful to Mr Kosuke Ito for explaining the derivation \eqref{10-24-3}. 
The authors are grateful to the referees.
HT was partially supported by the Grants-in-Aid for Japan Society for Promotion of Science (JSPS) Fellows (Grant No. 24.8116).
MH is partially supported by a MEXT Grant-in-Aid for Scientific Research (A) No. 23246071
and the National Institute of Information and Communication Technology (NICT), Japan.
The Centre for Quantum Technologies is funded by the
Singapore Ministry of Education and the National Research Foundation
as part of the Research Centres of Excellence programme. 

\appendix

\section{Organization of Appendix}
Before starting Appendix, we briefly explain the organization of Appendix.
In Appendix \ref{s5-s}, we give a similar discussion to Section \ref{s5} of the main body
based on the fidelity instead of the entropic quantity.
Although  the fidelity requires more complicated discussions, 
it brings tighter evaluation for the coherence of the dynamics.
Further, using the discussion based on the fidelity, we give a proof of 
Theorem \ref{th11-23} of the main body.
Since this discussion needs several technical lemmas, 
they are shown in Appendix \ref{as1}.

In Appendix \ref{s4-3}, we extend the discussion for shift-invariant model 
in Section \ref{s4b}
to the case with non-lattice Hamiltonian. 
In Appendix \ref{s6}, 
we discuss the relation between classical work extraction
and quantum work extraction.
Indeed, the model of quantum work extraction is more complicated than 
that of classical work extraction.
This difficulty is mainly caused by the effect of coherence.
Hence, 
recently some people \cite{ULK,LKJR} arose how the classical model is close to the quantum model.
Since this problem is not directly related to the main issue of this paper nevertheless its importance,
Appendix \ref{s6} of the supplement gives the answer to this question in the non-asymptotic setting.

\section{Approximation to internal unitary}\Label{s5-s}
\subsection{Approximation and coherence}\Label{s5a-s}
In this section, 
as another measure of coherence,
we focus on the entanglement fidelity:
\begin{align}
F_e(\Lambda, U_I, \rho_I )^2:=
\langle \Phi | U_I^\dagger 
\Lambda (|\Phi \rangle \langle \Phi |)
U_I|\Phi \rangle.
\end{align}
When the initial state $|\psi_a\rangle$ is generated with the probability $P_A(a)$,
namely when $\rho_{I}=\sum_{a}P_{A}(a)|\psi_{I,a}\rangle\langle\psi_{I,a}|$ holds,
the entanglement fidelity $F_e(\Lambda, U_I, \rho_I )$
characterizes any average fidelity as
\begin{align}
F_e(\Lambda, U_I, \rho_I )^2 
\le 
\sum_a P_{A}(a) \langle \psi_a| U_I^\dagger \Lambda(|\psi_a\rangle \langle \psi_a|)  U_I |\psi_a\rangle .
\end{align}
Since this value is zero in the ideal case,
this value can be regarded as the amount of the disturbance by the time-evolution $\Lambda$.
This quantity satisfies the quantum Fano inequality \cite{Schumacher},\cite[(8.51)]{23}
\begin{align}
&S_e(\Lambda,\rho_I)
= S_e(\Lambda_{U_I^\dagger}\circ \Lambda,\rho_I)\nonumber\\
&\le
h( F_e(\Lambda, U_I, \rho_I )^2 )
+
(1-F_e(\Lambda, U_I, \rho_I )^2 )
\log (d_I^2-1),\Label{3-15-12}
\end{align}
where $\Lambda_{U_I^\dagger}(\rho):= U_I^\dagger \rho U_I$.
Since this value is also zero in the ideal case,
we consider that this value is another measure of the disturbance by the time-evolution $\Lambda$.

\subsection{Correlation with external system}\Label{s5b-s}
Next, we discuss the correlation with external system by using the fidelity.
Now, we notice another expression of $I_{\rho''_{ZE}}(Z;E)$ as
\begin{align}
I_{\rho''_{ZE}}(Z;E) 
=&
\min_{\sigma_{E}}
D(\rho''_{ZE} \|\rho_{Z}'' \otimes \sigma_{E} ).
\end{align}
Following this expression, we consider the fidelity-type mutual information as follows
\begin{align}
I_{F,\rho''_{ZE}}(Z;E) 
:=&
-\log \max_{\sigma_{E}}
F(\rho''_{ZE} ,\rho_{Z}'' \otimes \sigma_{E} ).
\end{align}
We can show the following theorem.
\begin{theorem}\Label{3-16-2T}
The relations
\begin{align}
{I}_{F,\rho''_{ZE}}(Z;E) 
&=
-\log \sum_{z,z}P_Z(z)P_Z(z') F(\rho_{E|Z=z}'',\rho_{E|Z=z'}'')\nonumber\\
&\le S_2(P_Z)\Label{3-11-2eq}
\end{align}
hold, where $S_2(P_Z):= - \log \sum_z P_Z(z)^2$.
The equality holds when 
the states $\{ \rho_{E|Z=z}''\}_z$
are distinguishable.
\end{theorem}
Theorem \ref{3-16-2T} follows from 
a more general argument, Lemma \ref{3-16-7L} in Appendix \ref{as1} 

Since the equality in \eqref{3-11-2eq} holds in the ideal case,
we introduce another measure of the imperfectness of the correlation for the decrease of the energy as
\begin{align}
\Delta {I}_{F,\rho''_{ZE}}(Z;E) 
:= S_2(P_Z)- {I}_{F,\rho''_{ZE}}(Z;E) .
\end{align}

\begin{lemma}
Let ${\cal F}=({\cal H}_{E}, \hat{H}_{E}, U, \rho_{E})$
be an FQ-work extraction.
Assume that $\rho_I$ is commutative with $\hat{H}_I$.
When the CP-work extraction $CP({\cal F})$ is a level-4 CP-work extraction,
the states $\{ \rho_{E|Z=z}''\}_z$
are distinguishable.
\end{lemma}

\begin{proof}
Let $w_0$ be the eigenvalue of the Hamiltonian associated with the state $\rho_{E}$.
Then,
\begin{align}
\Tr P_{w_0+z'} \rho_{E|Z=z}''=\delta_{z,z'}.
\end{align}
Hence, the states $\{ \rho_{E|Z=z}''\}_z$ are distinguishable.
\end{proof}

\begin{remark}
We should remark the relation between 
the fidelity-type mutual information $I_{F,\rho''_{ZE}}(Z;E) $
and an existing mutual information measure.
Recently, a new type of quantum R\'{e}nyi relative entropy
$\tilde{D}_\alpha(\tau\|\sigma)$ was introduced \cite{WWY,MDSFT}.
When the order is $\frac{1}{2}$, it is written as
\begin{align}
\tilde{D}_{\frac{1}{2}}(\tau\|\sigma)
=-2 \log F(\tau,\sigma).
\end{align}
Using this relation,
the papers \cite{GW,Beigi,WWY} introduced 
the quantity $\tilde{I}_{\alpha,\rho}(Z;E)$ with general order $\alpha$
by 
\begin{align}
\tilde{I}_{\alpha,\rho}(Z;E)
:=
\min_{\sigma_E} \tilde{D}_{\alpha}(\rho\|\rho_Z \otimes \sigma_E),
\end{align}
whose operational meaning was clarified by \cite{HT}. 
So, our fidelity-type mutual information $I_{F,\rho}(Z;E) $
is written as
\begin{align}
2 I_{F,\rho}(Z;E)= \tilde{I}_{\frac{1}{2},\rho}(Z;E).
\end{align}
\end{remark}

\subsection{Trade-off with imperfectness of correlation}\Label{s5c-s}
We derive the fidelity version of the trade-off relation with imperfectness of correlation.
To give this kind of trade-off relation, we
prepare the isometry $V_{IR Z}$ from $\cH_{I}\otimes \cH_R$
to $\cH_{I}\otimes \cH_R\otimes \cH_Z$:
\begin{align}
V_{IR Z}:=\sum_{z} |z\rangle \otimes F_z.
\end{align}
Then, 
we define the distribution 
\begin{align}
\tilde{P}_Z(z):= \langle \Phi| U_I^\dagger F_z U_I |\Phi \rangle ,\Label{3-15-11eq}
\end{align}
and
the pure states 
$|\psi_{IRE|Z=z}''\rangle$ 
and $|\tilde{\psi}_{IR|Z=z}''\rangle$ 
by
\begin{align}
V_{IR Z} U_I |\Phi\rangle &= 
\sum_{z}\sqrt{\tilde{P}_Z(z)} | \tilde{\psi}_{IR|Z=z},z\rangle \\
V_{IR Z} U |\Phi,\psi_{E}\rangle &= 
\sum_{z}\sqrt{P_Z(z)} |\psi_{IRE|Z=z}'',z\rangle.
\end{align}

\begin{theorem}\Label{3-15-3T}
The quality of approximation
$F_e(\Lambda, U_I,\rho_I)$ and 
the amount of imperfectness of correlation 
$\Delta {I}_{F,\rho''_{ZE}}(Z;E)$
satisfy the following trade-off relation:
\begin{align}
-\log F_e(\Lambda, U_I,\rho_I)
+\Delta I_{F,\rho''_{ZE}}(Z;E)
& \ge S_2(P_Z)
\Label{3-8-10},
\end{align}
i.e.,
\begin{align}
-\log F_e(\Lambda, U_I,\rho_I)
\ge I_{F,\rho''_{ZE}}(Z;E)
\Label{3-8-10B}.
\end{align}
The equality holds if and only if $\tilde{P}_Z=P_Z$
and there exist pure states
$|\psi_{E|Z=z}''\rangle $
such that
$|\psi_{IRE|Z=z}''\rangle=
|\tilde{\psi}_{IR|Z=z}'',\psi_{E|Z=z}''\rangle$ 
and 
$\langle \psi_{E|Z=z'}''| \psi_{E|Z=z}''\rangle \ge 0$.
\begin{align}
&-\log F_e(\Lambda, U_I, \rho_I)
=
I_{F,\rho''_{ZE}}(Z;E) \nonumber\\
&=
-\log \sum_{z,z'}
P_Z(z)P_Z(z')
\langle \psi_{E|Z=z}| \psi_{E|Z=z'}\rangle
\Label{3-11-6}.
\end{align}
\end{theorem}

\begin{proof}
Since 
$F( U_I|\Phi\rangle \langle \Phi| U_I^\dagger ,
\rho_{IR}'')
=
F( V_{ZIR} U_I|\Phi\rangle \langle \Phi| U_I^\dagger V_{ZIR}^\dagger  ,
V_{ZIR}\rho_{IR}''V_{ZIR}^\dagger )$,
the inequality \eqref{3-8-10} is equivalent with 
\begin{align}
&F( V_{ZIR} U_I|\Phi\rangle \langle \Phi| U_I^\dagger 
V_{ZIR}^\dagger  ,
V_{ZIR} \rho_{IR}''V_{ZIR}^\dagger )\nonumber\\
&\le \max_{\sigma_E}
F(\rho_{ZE}'', \rho_{Z}'' \otimes \sigma_E ).
\end{align}
So, the desired argument follows from 
Lemma \ref{3-16-8L} in Appendix \ref{as1}. 
\end{proof}

Here, it is better to explain the difference between Theorem \ref{3-15-3T} and Theorem \ref{3-15-2T} of the main body.
Theorem \ref{3-15-2T} of the main body gives the relation between
the decoherence $S_e(\Lambda,\rho_I)$
and 
the imperfectness of correlation $\Delta {I}_{\rho''_{ZE}}(Z;E) $,
it does not require any relation with the internal unitary $U_I$.
To derive a relation with the approximation, 
we employ the quantum Fano inequality \eqref{3-15-12}.
However, Theorem \ref{3-15-3T} directly gives the relation between
the approximation $-\log F_e(\Lambda, U_I,\rho_I)$
and the imperfectness of correlation $\Delta {I}_{F,\rho''_{ZE}}(Z;E)$.
Hence, it requires the condition \eqref{3-15-11eq}.

\subsection{Shift-invariant model}\Label{s5d-s}
Next, we consider how to realize the case 
when the amount of the decoherence 
$-\log F_e(\Lambda, U_I,\rho_I)$
is close to zero,
and show Theorem \ref{th11-23} of the main body.
For simplicity, we consider this problem
under the shift-invariant model without an additional external system.
In the shift-invariant model, 
once we fix the energy-conservative unitary operator $F[U_I]$ on
$\cH_I \otimes \cH_E$ according to \eqref{3-12-10eq} of the main body,
the assumption of Theorem \ref{3-15-3T} holds 
and
the distribution $P_Z$ depends only on the initial state $\rho_I$ on the system $I$.
That is, $P_Z$ does not depend on the initial state $|\psi_{E}\rangle $ on the external system $E$.
In this case, we have
\begin{align}
|\psi_{E|Z=h_Ej}\rangle=
V^j|\psi_{E}\rangle.
\end{align}
In particular, when $|\psi_{E}\rangle=\sum_j \sqrt{P_J(j)}|j\rangle$,
the equality condition holds in both inequality \eqref{3-8-10}.
Theorem \ref{3-15-3T} implies that 
\begin{align}
&-\log F_e(\Lambda, U_I,\rho_I)
=
I_{F,\rho''_{ZE}}(Z;E) \nonumber\\
=&
-\log 
\sum_{z,z'}
P_Z(z)P_Z(z')
\sum_j \sqrt{P_J(j)}\sqrt{P_J(j+z-z')}
\Label{3-11-b}.
\end{align}
For example, when 
\begin{align}
P_J(j)&=
\left\{
\begin{array}{ll}
\frac{1}{2m+1} & \hbox{ if } |j| \le m \\
0 & \hbox{ otherwise.}
\end{array}
\right. \\
P_Z(h_E j)& =0  \hbox{ if } |j| \ge l,
\end{align}
we have
\begin{align}
&-\log F_e(\Lambda, U_I,\rho_I)
=
I_{F,\rho''_{ZE}}(Z;E) \nonumber\\
\le &
-\log (1-\frac{l}{2m+1})
\le 
\frac{l}{2m+1}.
\end{align}
Hence, applying \eqref{3-15-12}, 
we have
\begin{align}
S_e(\Lambda,\rho_I)
\le & h((1-\frac{l}{2m+1})^2)\nonumber \\
&+(1-(1-\frac{l}{2m+1})^2)
\log (d_I^2-1) \nonumber \\
=&
h\big(2\frac{l}{2m+1}- (\frac{l}{2m+1})^2\big)\nonumber \\
& +\big(2\frac{l}{2m+1}- (\frac{l}{2m+1})^2\big)
\log (d_I^2-1).
\end{align}
Hence, we obtain Theorem \ref{th11-23} of the main body.

\subsection{Trade-off under CP-work extraction}\Label{s5e-s}
In Appendix \ref{s5c-s}, we discuss the trade-off relation between the 
imperfectness of correlation and the quality of approximation
under a FQ-work extraction.
In this subsection, we discuss the trade-off relation under a CP-work extraction 
${\cal G}:=\{{\cal E}_j,w_j\}_j$ on the internal system ${\cal H}_I$ with the Hamiltonian $\hat{H}_I$.
To discuss the trade-off relation, in the internal system ${\cal H}_I$,
we consider the internal unitary $U_I$ to be approximated and
the initial mixture state $\rho_I$, which is assumed to be commutative with $\hat{H}_I$.
To qualify the approximation, we employ the measure
$F_e(\sum_j {\cal E}_j,U_I,\rho_I)$.
To evaluate the imperfectness of correlation, 
we consider the purification $|\Phi\rangle$ of $\rho_I$,
and introduce the joint distribution $P_{ZW}$ as
\begin{align}
P_{ZW}(z,w):=\sum_{j:w_j=w} \Tr_{R I} {\cal E}_j(|\Phi\rangle \langle \Phi|) E_z,
\end{align}
where the projection $E_z$ is defined in the same way as in Subsection \ref{s5b} of the main body.
Then, we employ the measure as
$\Delta I_{F,P_{ZW}}(Z;W)=S_2(P_Z) -I_{F,P_{ZW}}(Z;W)$.

As corollary of Theorem \ref{3-15-3T}, we have the following.
\begin{corollary} \Label{4-10-2C}
Given a CP-work extraction $\{{\cal E}_j,w_j\}_j$.
The quality of approximation
$F_e(\sum_j {\cal E}_j, U_I,\rho_I)$ and 
the amount of imperfectness of correlation 
$\Delta {I}_{F,P_{ZE}}(Z;E)$
satisfy the following trade-off relation:
\begin{align}
-\log F_e(\sum_j {\cal E}_j, U_I,\rho_I)
+\Delta I_{F,P_{ZE}}(Z;E)
& \ge S_2(P_Z)
\Label{3-8-10x}.
\end{align}
\end{corollary} 

\begin{proofof}{Corollary \ref{4-10-2C}}
Notice that Theorem \ref{3-15-3T} does not assume any energy conservation law.
Then, we take a Stinespring representation 
$({\cal H}_E, U, \rho_E)$ with a pure state $\rho_E$ of 
$\{{\cal E}_{j},w_{j}\}_{j\in {\cal J}}$ as follows.
\begin{align}
{\cal E}_{j}(\rho_I)=
\Tr_{E} U_{IE}
(\rho_I \otimes \rho_{E}) U_{IE}^\dagger 
(\hat{1}_{I} \otimes E_j),
\end{align}
where $\{E_j\}$ is a projection valued measure on ${\cal H}_E$.
Notice that $({\cal H}_E, U, \rho_E)$ is an FQ-work extraction.
Here, we do not care about whether the FQ-work extraction satisfies any energy conservation law.
Then, we apply Theorem \ref{3-15-3T}.
Since information processing inequality for the fidelity
yields that 
$I_{F,\rho''_{ZE}}(Z;E) \ge I_{F,P_{ZE}}(Z;E)$,
the relation \eqref{3-8-10} derives
\eqref{3-8-10x}.
\end{proofof}

Further, we have the following corollary.
\begin{corollary}
Assume that $\rho_I$ is commutative with $\hat{H}_I$.
For a level-4 CP-work extraction $\{{\cal E}_j,w_j\}_j$,
$\Delta I_{P_{ZE}}(Z;E) =\Delta I_{F,P_{ZE}}(Z;E) =0$.
So, we have
\begin{align}
-\log F_e(\sum_j {\cal E}_j, U_I,\rho_I)
& \ge S_2(P_Z)
\Label{3-8-10y}.
\end{align}
\end{corollary}
This corollary implies that the dynamics of a level-4 CP-work extraction is far
from any internal unitary.

\section{Shift-invariant model with non-lattice Hamiltonian}
\Label{s4-3}
Now, we show that our discussion for shift-invariant model in Section \ref{s4b}
can be extended to the case with non-lattice Hamiltonians $\hat{H}_I$.
In this case, 
we cannot employ the space $L^2(\mathbb{Z})$
for the non-degenerate external system $E1$.
One idea is to replace the space $L^2(\mathbb{Z})$
by the space $L^2(\mathbb{R})$.
However, in this method,
to satisfy the condition \eqref{3-5-1}
with the measurement of the Hamiltonian $\hat{H}_{E1}$,
we need to prepare the state whose wave function is a delta function.
To avoid such a mathematical difficulty, we employ another construction of the non-degenerate external system $E1$.

Let $\{h_i\}$ be the set of eigenvalues of $\hat{H}_I$.
We choose the set $\{h_{E,l}\}_{l=1}^L$ satisfying the following.
(1) When rational numbers $t_1, \ldots, t_L$ satisfy $\sum^{L}_{l=1} t_l h_{E,l}=0$, the equality  $t_l=0$ holds for all $l$.
(2) $\{h_i-h_j\}_{i,j} \subset \{ \sum^{L}_{l=1}n_{l} h_{E,l} \}_{n_{l} \in \mathbb{Z}}$.
Then, we choose $\cH_{E1}$ to be $L^2(\mathbb{Z})^{\otimes L}$
and the Hamiltonian $\hat{H}_{E1}$ to be
$\sum_{j_1, \ldots, j_L} 
h_{E,1}j_1+ \ldots + h_{E,L}j_L
|j_1, \ldots, j_L\rangle\langle j_1, \ldots, j_L|$.

As in the lattice case, the external system $E1$ does not have a ground state. 
Because $E1$ in the present case is composed of $L$ ladder systems, we can reconstruct the property of $E1$ in $L$ pairs of harmonic oscillators\cite[Section IV in Supplement]{catalyst}.

Then, a shift-invariant unitary is defined as follows.
We define $L$ displacement operators 
$V_{E1,l}:=\sum_{j}|j+1 \rangle_E ~_E\langle j| $
on the $l$-th space $L^2(\mathbb{Z})$.

\begin{definition}[Shift-invariant unitary]\Label{memoryless2}
A unitary $U$ on $\cH_I\otimes \cH_{E1}$ is called shift-invariant when
\begin{align}
U V_{E1,l}=V_{E1,l} U \Label{10-24}
\end{align}
for $l=1, \ldots, L$.
\end{definition}

Then, the definition of $F[U_I]$ is changed to 
\begin{align}
F[U_I]:=
\sum_{j_1, \ldots, j_L} 
(\otimes^{L}_{l=1}V_{E1,l}^{j_l})
W U_I W^\dagger 
(\otimes^{L}_{l=1}V_{E1,l}^{j_l})^\dagger
\Label{3-12-10eq2}.
\end{align}
Under this replacement, we have Lemma \ref{3-13-1L}.
Other definitions and lemmas 
in Section \ref{s4b} work with this replacement.


\section{Classical work extraction}\Label{s6}
\subsection{Formulation}\Label{s6a}
Here, we consider the relation with work extraction from a classical system, again.
Since the original Jarzynski's formulation \cite{Jarzynski}
considers only the deterministic time evolution, 
we need to give a formulation of probabilistic time evolution in the classical system
as an extension of Jarzynski's formulation.
That is, our formulation can be reduced into the probabilistic work extractions from the classical systems, which is defined as follows;
\begin{definition}[Classical work extraction]
We consider a classical system ${\cal X}$ and its Hamiltonian which is given as a real-valued function $h_X$ on ${\cal X}$.
We also consider a probabilistic dynamics $T(x|x')$ on ${\cal X}$, which is a probability transition matrix, i.e., $\sum_{x} T(x|x')=1$.
We refer to the triplet $({\cal X},h_X,T)$ as a classical work extraction.   
\end{definition}

This model includes the previous fully classical scenario\cite{Jarzynski,Croocks,Sagawa2010,Ponmurugan2010,Horowitz2010,Horowitz2011,Sagawa2012,Ito2013}, in which we extract work from classical systems.  
For example, the set up of the Jarzynski equality \cite{Jarzynski} is a special case that $T(x|x')$ is 
{\it invertible and deterministic},
i.e.,
$T(x|x')$ is given as $\delta_{x,f(x')}$ with an invertible function $f$.
We call such a transition matrix invertible and deterministic.
In this case,
the transition matrix $T$ is simply written as $f_*$.
That is, for a distribution $P$, we have
\begin{align}
f_*(P)(x)= P(f^{-1}(x)).\Label{4-19-4eq}
\end{align}
Notice that the relation 
\begin{align}
(g \circ f)_* (P) = g_*(  f_* (P))
\end{align}
holds.
When the transition matrix $T$ is {\it bi-stochastic}, i.e., $\sum_x T(x|x')=1$,
there exist a set of invertible functions $f_l$ and a distribution $P_L(l)$ such that $T=\sum_l P_L(l) f_{l*}$, i.e., 
\begin{align}
T(x|x')=\sum_l P_L(l) \delta_{x,f_l(x')}.\Label{3-13-1eq}
\end{align}
This relation means that 
a bi-stochastic transition matrix $T$
can be realized by a randomized combination of invertible and deterministic dynamics.

Under a classical work extraction $({\cal X},h_X,T)$,
because of the energy conservation,
the amount of the extracted work is given as 
\begin{align}
w_{x,x'}=h_X(x)-h_X(x') \Label{3-13-2eq}
\end{align}
when the initial and final states are $x$ and $x'$. 
More generally,
the initial state is given as a probability distribution $P_X$ on ${\cal X}$.
In this case, the amount of the extracted work is $w$
with the probability
\begin{align}
\sum_{x',x: w= h_{X}(x)-h_{X}(x')} P(x)T(x'|x).
\Label{3-13-2eqb}
\end{align}

For the entropy of the amount of extracted work, 
we can show the following lemma in the same way as Lemma \ref{4-8-1L} of the main body.
\begin{lemma}\Label{4-8-5L}
We denote the random variable describing the amount of extracted work by $W$. 
The entropy of $W$ is evaluated as
\begin{align}
S[W] \le 2 \log N,\Label{4-8-1eqx}
\end{align}
where $N$ is the number of elements of the set 
$ \{h_X(x)\}_{x \in {\cal X}}$.
\end{lemma}

Indeed, our CP-work extraction 
decides the amount of extracted work probabilistically dependently of 
the initial state on the internal system.
So, our CP-work extraction can be regarded as a natural quantum extension of
a probabilistic classical work extraction.

\subsection{Relation with CP-work extraction}\Label{s6b}
Next, we discuss the relation between CP-work extractions and classical work extractions. 
A CP-work extraction can be converted to a classical work extraction 
under a suitable condition as follows.
\begin{definition}[classical description]\Label{cdes1}
For a level-4 CP-work extraction $\{{\cal E}_{j}, w_{j}\}$,
we define the probability transition matrix
$T_{\{{\cal E}_{j}, w_{j}\}}$ as
\begin{align}
T_{\{{\cal E}_{j}, w_{j}\}}(y|x):=
\sum_{j}
\langle y|{\cal E}_{j}(\Pi_{x})|y\rangle ,\label{definitionT}
\end{align}
and the function $h_X(x)$ is given as the eigenvalue of the Hamiltonian $\hat{H}_I$ associated with the eigenstate $|x\rangle$.
Then, we refer to the triplet
$({\cal X},h,T_{\{{\cal E}_{j},w_{j}\}})$
as the classical description of the level-4 CP-work extraction $\{{\cal E}_{j}, w_{j}\}$,
and denote it by ${\cal T}(\{{\cal E}_{j}, w_{j}\})$.
\end{definition}

Hence, when the support of the initial state $\rho_I$ of 
a FQ-work extraction 
${\cal F}=({\cal H}_{E}, \hat{H}_{E}, U, \rho_{E})$ 
belongs to an eigenspace of the Hamiltonian $\hat{H}_{E}$,
its classical description is given as
${\cal T}(CP({\cal F}))$.
The above classical description gives the behavior of 
a given level-4 CP-work extraction $\{{\cal E}_{j}, w_{j}\}$.
When the initial state on the internal system $I$
is the eigenstate $|x\rangle$,
due to the condition \eqref{CP2} of the main body,
the amount of the extracted  work is $w$ with the probability
\begin{align}
\sum_{j:w=w_{j}}\sum_{y:w_j=h_x-h_y}
\langle y|{\cal E}_{j}(\Pi_{x})|y\rangle 
=
\sum_{y:w=h_x-h_y}\!\!
T_{\{{\cal E}_{j}, w_{j}\}}(y|x).
\end{align}
Hence, the classical description 
${\cal T}(\{{\cal E}_{j}, w_{j}\})$
gives the stochastic behavior in this case.
More generally, we have the following theorem.

\begin{theorem}\Label{Th4}
Assume that ${\cal P}_{\hat{H}_I}(\rho_I)$ is written as 
$\sum_{x} P_X(x)\Pi_{x}$.
(For the definition of ${\cal P}_{\hat{H}_I}(\rho_I)$, see \eqref{4-8-9eqx} of the main body.)
When we apply a level-4 CP-work extraction $\{{\cal E}_{j}, w_{j}\}$
to the system $I$ with the initial state $\rho_I$,
the amount of the extracted work is $w$ with the probability
\begin{align}
\sum_{x,y:w=h_x-h_y}
T_{\{{\cal E}_{j}, w_{j}\}}(y|x) P_X(x).
\end{align}
\end{theorem}

\begin{proof}
Due to Lemma \ref{3-14-1L} of the main body,
the probability of the amount of the extracted work $w$ is
\begin{align}
&\sum_{j:w_{j}=w}\Tr {\cal E}_{j}(\rho)=\sum_{j}\Tr {\cal P}_{\hat{H}_I}({\cal E}_{j}(\rho))\delta_{w,w_{j}}\nonumber\\
=&\sum_{j}\Tr {\cal E}_{j}({\cal P}_{\hat{H}_I}(\rho))\delta_{w,w_{j}}
\nonumber\\
=&
\sum_{j}\sum_{x,y} P_X(x)\langle y| {\cal E}_{j}(\Pi_{x})|y\rangle\delta_{w,w_{j}}\nonumber\\
\stackrel{(a)}{=}&\sum_{j}\sum_{x,y} P_X(x)\langle y| {\cal E}_{j}(\Pi_{x})|y\rangle\delta_{w,w_{j}}\delta_{w_{j},h_{x}-h_{y}} \nonumber\\
=&
\sum_{j}\sum_{x,y} P_X(x)\langle y| {\cal E}_{j}(\Pi_{x})|y\rangle\delta_{w,h_{x}-h_{y}}\delta_{w_{j},h_{x}-h_{y}}\nonumber\\
\stackrel{(b)}{=}&\sum_{j}\sum_{x,y} P_X(x)\langle y| {\cal E}_{j}(\Pi_{x})|y\rangle\delta_{w,h_{x}-h_{y}}
\nonumber\\
\stackrel{(c)}{=} &
\sum_{x,y:w_j=h_x-h_y}
T_{\{{\cal E}_{j}, w_{j}\}}(y|x) P_X(x).
\end{align}
where $(a)$ and $(b)$ follow from \eqref{CP2} of the main body, and $(c)$ follows from \eqref{definitionT}.
\end{proof}

Due to this theorem, in order to discuss the amount of extracted work in the level-4 CP-work extraction,
it is sufficient to handle the classical description.
Theorem \ref{Th4} is written as a general form and contains the case when 
the initial state $\rho_I$ is commutative with the Hamiltonian.
In this commutative case, 
the amount of extracted work can be simulated by the classical model.

\begin{lemma}
Given a level-4 CP-work extraction $\{{\cal E}_{j}, w_{j}\}$,
when the CP-work extraction $\{{\cal E}_{j}, w_{j}\}$ is unital,
the transition matrix $T_{\{{\cal E}_{j}, w_{j}\}} $
is bi-stochastic. 

\end{lemma}

\begin{proof}
Since
\begin{align}
\sum_j 
\langle y|{\cal E}_{j}(\frac{\hat{1}}{|{\cal X}|})|y\rangle 
=
T_{\{{\cal E}_{j}, w_{j}\}}(y|x) 
\frac{1}{|{\cal X}|},
\end{align}
when the CP-work extraction $\{{\cal E}_{j}, w_{j}\}$ is unital,
the transition matrix $T_{\{{\cal E}_{j}, w_{j}\}} $
is bi-stochastic. 
\end{proof}

\begin{lemma}\Label{3-17-2L}
Given a classical work extraction $({\cal X},h_X,T)$,
the transition matrix $T$ is bi-stochastic 
if and only if 
there exists a standard FQ-work extraction ${\cal F}$ such that
$({\cal X},h_X,T)= {\cal T}(CP({\cal F}))$.
\end{lemma}

Lemma \ref{3-17-2L} will be shown after Lemma \ref{3-17-3L}.
Since the set of standard FQ-work extractions
is considered as the set of preferable work extraction,
it is sufficient to optimize the performance 
under the set of classical work extractions with a bi-stochastic transition matrix.
That is, both models yield the same distribution of the amount of extracted work. 
So, our model can be applied to the discussions for the tail probability and the variance 
for the amount of work extraction as well as the expectation.

As a subclass of bi-stochastic matrices, we consider the set of uni-stochastic matrices. 
A bi-stochastic matrix $T$ is called {\it uni-stochastic} 
when there exists a unitary matrix $U$ such that
$T(x|x')=|U_{x,x'}|^2$.
According to the discussion in Section \ref{s4b} and Appendix \ref{s4-3} of the main body,
we can consider a FQ-work extraction ${\cal F}=({\cal H}_{E1}, \hat{H}_E, F[U_I], \rho_E) $, where $\rho_E$ is a pure eigenstate of $\hat{H}_E$.
As mentioned in the end of Section \ref{s4b} of the main body,
since the corresponding CP-work extraction $CP({\cal F})$ depends only on 
the internal unitary $U_I$,
the CP-work extraction is denoted by $\hat{CP}(U_I)$.
Then, we have the following lemma.

\begin{lemma}\Label{3-17-3L}
Given a classical work extraction $({\cal X},h_X,T)$,
when 
the transition matrix $T$ is a uni-stochastic matrix
satisfying $T(x|x')=|U_{I;x,x'}|^2$ with a internal unitary $U_I$
then 
\begin{align}
({\cal X},h_X,T)= {\cal T}(\hat{CP}(U_I)).
\Label{4-8-8eq}
\end{align}
In the corresponding FQ-work extraction ${\cal F}=({\cal H}_{E1}, \hat{H}_E, F[U_I], \rho_E)$,
the entropy of the final state of external system is given as
\begin{align}
S[W] \ge S(\Tr_I F[U_I] (\rho_I\otimes \rho_E) F[U_I]^\dagger )
\Label{4-8-9eq}
\end{align}
for any pure eigenstate $\rho_E$ of $\hat{H}_E$.
\end{lemma}

\begin{proof}
The relation \eqref{4-8-8eq} can be shown from the definition of $\hat{CP}(U_I)$.
The relation \eqref{4-8-9eq} can be shown in the same way as \eqref{4-8-2L} of the main body.
\end{proof}

\begin{proofof}{Lemma \ref{3-17-2L}}
Given a bi-stochastic matrix $T$, there exist a probability distribution $P_A(a)$ and a unitary matrix $U_{I;a}$ such that
\begin{align}
T(y|x)=\sum_{u}P_A(a) |U_{I;a;x,y}|^2.
\end{align}
Then, we choose the fully degenerate system $\cH_{E2}$ spanned by 
$\{|a \rangle_{E2}\}$
and the initial state 
$\rho_{E2}:=\sum_a P_A(a)|a \rangle_{E2}~_{E2}\langle a|$.
We define the unitary 
$U:= \sum_a F[U_{I;a}] \otimes |a \rangle_{E2}~_{E2}\langle a|$.
So, we have
$({\cal X},h_X,T)= {\cal T}(
\cH_{E},\hat{H}_E,U,\rho_{E1}\otimes \rho_{E2})$
for any pure state $\rho_{E1}$ on $\cH_{E1}$.
\end{proofof}

As a special case of Lemma \ref{3-17-3L}, we have the following lemma. 
\begin{lemma}\Label{4-2-2L}
Given a classical work extraction $({\cal X},h_X,f_*)$ with a invertible function $f$,
the unitary 
\begin{align}
U_f:|x\rangle \mapsto |f(x)\rangle
\end{align}
satisfies $({\cal X},h_X,T)= {\cal T}({\cal H}_{E1}, \hat{H}_E, F[U_f], \rho_E)$.
\end{lemma}

Therefore, any invertible and deterministic transition matrix $T$ 
can be simulated by a shift-invariant FQ-work extraction 
only with the non-degenerate external system.
So, the reduction to classical work extraction
will be helpful to analyze the heat engine.
That is, in the several settings, the analysis of heat engine can be essentially reduced to the analysis of classical work extraction.

\section{Relations among fidelities on tripartite system}\Label{as1}
In this appendix, we derive several useful relations 
among fidelities on a tripartite system $\cH_A,\cH_B,\cH_C$.
We consider 
the state 
$|\Psi\rangle:=
\sum_{a} \sqrt{\tilde{P}_A(a)}|a, \psi_{B|a}\rangle$
on $\cH_A \otimes \cH_B$,
and
the state 
$|\Phi\rangle:=
\sum_{a} \sqrt{{P}_A(a)}|a, \phi_{BC|a}\rangle$
on $\cH_A \otimes \cH_B\otimes \cH_C$,
Then, we denote $\rho:= |\Phi\rangle \langle\Phi|$.
We also define
$|\phi_{C|a}\rangle
:=\langle \psi_{B|a}| \phi_{BC|a}\rangle$.

In this case, we have the following lemma.
\begin{lemma}\Label{3-16-1L}
\begin{align}
&F(|\Psi\rangle \langle\Psi| ,\rho_{AB})\nonumber\\
&=
\sum_{a,a'} 
\sqrt{\tilde{P}_A(a) P_A(a)}
\sqrt{\tilde{P}_A(a') P_A(a')}
\langle \phi_{C|a'}| \phi_{C|a}\rangle 
\end{align}
\end{lemma}
 
\begin{proof}
\begin{align}
&F(|\Psi\rangle \langle\Psi| ,\rho_{AB}) \nonumber\\
=& \max_{|{\psi}_C\rangle}
F(|\Psi\rangle \langle\Psi|
|{\psi}_C\rangle \langle{\psi}_C|, 
|\Phi\rangle \langle\Phi|
)^2 \nonumber\\
=&\max_{|{\psi}_C\rangle}
|\sum_z 
\sqrt{\tilde{P}_A(a) P_A(a)}
\langle {\psi}_{B|a},{\psi}_C|
\phi_{BC|a}\rangle| \nonumber\\
=&
\max_{|{\psi}_C\rangle}
|\sum_z 
\sqrt{\tilde{P}_A(a) P_A(a)}
| \phi_{C|a}\rangle| \nonumber\nonumber\\
=&
\sum_{a,a'} 
\sqrt{\tilde{P}_A(a) P_A(a)}
\sqrt{\tilde{P}_A(a') P_A(a')}
\langle \phi_{C|a'}| \phi_{C|a}\rangle 
\end{align}
\end{proof}

\begin{lemma}\Label{3-16-7L}
When $\rho_{AC}$ is written as
$\sum_{a}P_A(a) |a\rangle \langle a| \otimes \rho_{C|a}$, 
we have
\begin{align}
\max_{\sigma_C} F(\rho_{AC}, \rho_A \otimes \sigma_C)
=\sum_{a,a'}P_A(a) P_A(a') F(\rho_{C|a},\rho_{C|a'})\Label{3-16-4eq}.
\end{align}
\end{lemma}

\begin{proof}
We firstly show the case when the state $\rho_{C|a}$
is a pure state $|\phi_{C|a}\rangle$.
We choose 
the purification 
$|\Phi(\{e^{i\theta_a}\})\rangle:=
\sum_{a} e^{i\theta_a} \sqrt{\tilde{P}_A(a)}|a, a, \phi_{C|a}\rangle$
of $\rho_{AC}$
on $\cH_A \otimes \cH_B\otimes \cH_C$,
and the purification 
$|\Psi(\{e^{i\theta_a'}\})\rangle:=
\sum_{a} e^{i\theta_a'} \sqrt{\tilde{P}_A(a)}|a, a\rangle$
of $\rho_{A}$ on $\cH_A \otimes \cH_B$,
Applying Lemma \ref{3-16-1L},
we have
\begin{align}
& \max_{\sigma_C} F(\rho_{AC}, \rho_A \otimes \sigma_C) \nonumber\\
=
& \max_{\sigma_C}
\max_{\{e^{i\theta_a}\}, \{e^{i\theta_a'}\}}
F(\Pi_{\Psi(\{e^{i\theta_a'}\})} 
\otimes \sigma_C,
\Pi_{\Phi(\{e^{i\theta_a}\})}) \nonumber\\
=
& 
\max_{\{e^{i\theta_a}\}, \{e^{i\theta_a'}\}}
\max_{\sigma_C}
F(\Pi_{\Psi(\{e^{i\theta_a'}\})} 
\otimes \sigma_C,
\Pi_{\Phi(\{e^{i\theta_a}\})}) \nonumber\\
=&
\max_{\{e^{i\theta_a}\}, \{e^{i\theta_a'}\}}
F(\Pi_{\Psi(\{e^{i\theta_a'}\})},
\Tr_C
\Pi_{\Phi(\{e^{i\theta_a}\})}
)\nonumber\\
=&
\max_{\{e^{i\theta_a}\}, \{e^{i\theta_a'}\}}
\sum_{a,a'} 
e^{i(\theta_a-\theta_a')-i(\theta_a'-\theta_a)}
P_A(a) P_A(a')
\nonumber\\
& \cdot
\langle \phi_{C|a'}| \phi_{C|a}\rangle \nonumber\\
=&
\sum_{a,a'} 
e^{i(\theta_a-\theta_a')-i(\theta_a'-\theta_a)}
P_A(a) P_A(a')
|\langle \phi_{C|a'}| \phi_{C|a}\rangle| ,\Label{a4}
\end{align}
where we use the abbreviations
\begin{align}
\Pi_{\Psi(\{e^{i\theta_a'}\})}&:=|\Psi(\{e^{i\theta_a'}\})\rangle \langle\Psi(\{e^{i\theta_a'}\})|\\
\Pi_{\Phi(\{e^{i\theta_a}\})}&:=|\Phi(\{e^{i\theta_a}\})\rangle \langle\Phi(\{e^{i\theta_a}\})|
\end{align}
The equality \eqref{a4} implies \eqref{3-16-4eq}.

Now, we going to the general case.
We fix a purification $|\phi_{CD|a}\rangle $
of $\rho_{C|a}$ on $\cH_C\otimes \cH_D$
so that
$F(\rho_{C|a},\rho_{C|a'})
=|\langle \phi_{C|a'}| \phi_{C|a}\rangle| $.
We choose 
the purification 
$|\Phi(\{e^{i\theta_a}\})\rangle:=
\sum_{a} e^{i\theta_a} \sqrt{\tilde{P}_A(a)}|a, a, \phi_{CD|a}\rangle$
of $\rho_{AC}$
on $\cH_A \otimes \cH_B\otimes \cH_C\otimes \cH_D$,
and the purification 
$|\Psi(\{e^{i\theta_a'}\})\rangle:=
\sum_{a} e^{i\theta_a'} \sqrt{\tilde{P}_A(a)}|a, a\rangle$
of $\rho_{A}$ on $\cH_A \otimes \cH_B$.
Similarly, we have
\begin{align}
& \max_{\sigma_C} F(\rho_{AC}, \rho_A \otimes \sigma_C) \nonumber\\
=
& \max_{\sigma_{CD}}
\max_{\{e^{i\theta_a}\}, \{e^{i\theta_a'}\}}
F(\Pi_{\Psi(\{e^{i\theta_a'}\})} 
\otimes \sigma_{CD},
\Pi_{\Phi(\{e^{i\theta_a}\})}) \nonumber\\
=
& 
\max_{\{e^{i\theta_a}\}, \{e^{i\theta_a'}\}}
\max_{\sigma_{CD}}
F(\Pi_{\Psi(\{e^{i\theta_a'}\})} 
\otimes \sigma_{CD},
\Pi_{\Phi(\{e^{i\theta_a}\})}) \nonumber\\
=&
\sum_{a,a'} 
e^{i(\theta_a-\theta_a')-i(\theta_a'-\theta_a)}
P_A(a) P_A(a')
|\langle \phi_{C|a'}| \phi_{C|a}\rangle| \nonumber\\
=&
\sum_{a,a'} 
e^{i(\theta_a-\theta_a')-i(\theta_a'-\theta_a)}
P_A(a) P_A(a')
F(\rho_{C|a},\rho_{C|a'}),
\end{align}
which implies \eqref{3-16-4eq}.
\end{proof}

\begin{lemma}\Label{3-16-8L}
When $P_A=\tilde{P}_A$,
\begin{align}
F(|\Phi\rangle \langle \Phi|, \rho_{AB} )
\le 
\max_{\sigma_C} F(\rho_{AC}, \rho_A \otimes \sigma_C)
\Label{3-8-10C}.
\end{align}
The equality hold 
if and only if 
$\langle \phi_{C|a'} |\phi_{C|a}\rangle \ge 0$
and
$\langle \phi_{C|a} |\phi_{C|a}\rangle =1$.
\end{lemma}

\begin{proof}
We use the notations in Lemmas \ref{3-16-1L} and \ref{3-16-7L}.
Then, we have
\begin{align}
\langle \phi_{C|a'}| \phi_{C|a}\rangle 
+
\langle \phi_{C|a}| \phi_{C|a'}\rangle 
\le
2 F(\rho_{C|a},\rho_{C|a'}).\Label{3-16-9eq}
\end{align}
Combining Lemmas \ref{3-16-1L} and \ref{3-16-7L}, we have
\begin{align}
&F(|\Phi\rangle \langle \Phi|, \rho_{AB} ) \nonumber\\
=&\sum_{a,a'} 
\sqrt{\tilde{P}_A(a) P_A(a)}
\sqrt{\tilde{P}_A(a') P_A(a')}
\langle \phi_{C|a'}| \phi_{C|a}\rangle \nonumber\\
\le &
\sum_{a,a'}P_A(a) P_A(a') F(\rho_{C|a},\rho_{C|a'}) \nonumber\\
=&
\max_{\sigma_C} F(\rho_{AC}, \rho_A \otimes \sigma_C).
\end{align}
Hence, we obtain \eqref{3-8-10C}.
The equality in \eqref{3-8-10C} holds if and only if
that in \eqref{3-16-9eq} holds.
So, we obtain the desired equivalence.
\end{proof}

\renewcommand{\refname}{\vspace{-1cm}}

\end{document}